\documentclass[aps,pra,superscriptaddress,groupedaddress,bibnotes,twocolumn,reprint]{revtex4-2}
\usepackage{graphicx}
\usepackage{dsfont}
\usepackage{amssymb}
\usepackage{amsthm}
\usepackage{braket}
\usepackage{tikz}
\usepackage{amsmath}
\usetikzlibrary{positioning,automata}
\usepackage{ragged2e}
\usepackage{marvosym}

\usepackage[caption=false,labelformat=simple,labelfont={normalsize}]{subfig}

\newtheoremstyle{bfnote}%
  {}{}
  {\itshape}{}
  {\bfseries}{.}
  { }{\thmname{#1}\thmnumber{ #2}\thmnote{ (#3)}}

\theoremstyle{bfnote}

\newtheorem{theorem}{Theorem}

\newtheorem{proposition}[theorem]{Proposition}
\newtheorem{definition}[theorem]{Definition}
\newtheorem*{remark}{Remark}
\usepackage{xcolor,hyperref}
\usepackage{times,txfonts}

\usepackage{blindtext}
\usepackage{geometry}
\newcommand{\lscr}{\underline{\mathcal{R}_{\mathcal{F}}}(\rho)}
\newcommand{\rob}{\mathcal{R}_{\mathcal{F}}(\rho)}
\newcommand{\lscrm}{\underline{\mathcal{R}^{(m)}_{\mathcal{F}}}(\rho)}
\newcommand{\robm}{\mathcal{R}^{(m)}_{\mathcal{F}}(\rho)}
\newcommand{\ff}{\mathcal{F}}
\newcommand{\xiseq}{\{\xi_n\}_n}
\newcommand{\ddh}{\mathcal{D(H)}}
\newcommand{\tth}{\mathcal{T(H)}}
\newcommand{\bbh}{\mathcal{B(H)}}

\DeclareMathOperator{\tr}{Tr}
\DeclareMathOperator{\cone}{cone}
\newcommand{\kkh}{\mathcal{K(H)}}
\DeclareMathOperator{\povm}{POVM}
\newcommand{\swap}{V_{\textup{SWAP}}}

\begin{document}

\title{All non-Gaussian states are advantageous for channel discrimination: \\ Robustness of non-convex continuous variable quantum resources}

\author{Leah Turner}
\email{pmylt1@nottingham.ac.uk}
\affiliation{School of Mathematical Sciences and Centre for the Mathematical and Theoretical Physics of Quantum Non-Equilibrium Systems, University of Nottingham, University Park, Nottingham, NG7 2RD, United Kingdom}
 
\author{M\u{a}d\u{a}lin Gu\c{t}\u{a}}
\email{madalin.guta@nottingham.ac.uk}
\affiliation{School of Mathematical Sciences and Centre for the Mathematical and Theoretical Physics of Quantum Non-Equilibrium Systems, University of Nottingham, University Park, Nottingham, NG7 2RD, United Kingdom}

\author{Gerardo Adesso\textsuperscript{\ \!\Letter\ }}
\email{gerardo.adesso@nottingham.ac.uk}
\affiliation{School of Mathematical Sciences and Centre for the Mathematical and Theoretical Physics of Quantum Non-Equilibrium Systems, University of Nottingham, University Park, Nottingham, NG7 2RD, United Kingdom}

\author{\vspace*{0.1cm}}
\affiliation{\textrm{\em{\normalsize{\textsuperscript{\Letter}}}Author to whom any correspondence should be addressed}}

\begin{abstract}
Which quantum phenomena are advantageous for information processing tasks? By classifying quantum states and operations as resourceful versus non-resourceful, or free, the mathematical formalism of  quantum resource theories helps to address such questions. 
For the task of discriminating channels applied to a probe state, it has been shown that under certain conditions -- namely, the set of free states being convex in any dimension or possibly non-convex but finite dimensional -- every resourceful probe state can provide an advantage,  quantified by a resource monotone known as generalized robustness. In this work, bypassing the limitations of previous studies, we define the generalized robustness for an arbitrary resource theory and show that it admits two operational interpretations. Firstly, it provides an upper bound on the maximal advantage in channel discrimination tasks implemented on multiple copies of the probe states. Secondly, in many physically relevant theories, it quantifies the advantage in single-copy channel discrimination tasks in a worst case scenario. % when considering a decomposition of the free states into convex subsets. 
We further present a general construction of multi-copy resource witnesses and provide practical methods to bound the generalized robustness in experiments. Finally, we apply our results to a key case study in continuous variable quantum information: the resource theory of non-Gaussianity. This theory is naturally defined by a set of free states (Gaussian states) that is non-convex and infinite-dimensional. Our work then shows conclusively that all non-Gaussian states can provide an advantage in some channel discrimination task, even those that are simply mixtures of Gaussian states and are typically disregarded for other quantum information tasks. To illustrate our findings, we provide exact formulas for the robustness of non-Gaussianity of Fock states, along with an analysis of the robustness for a family of non-Gaussian states within the convex hull of Gaussian states.

\end{abstract}

\maketitle

\footnotetext[1]{Corresponding author.}

\section{Introduction}\label{sec:Intro} 

The integration of quantum mechanics into modern science and technology has led to the development of quantum technologies, which hold the potential to outperform their classical counterparts in numerous applications. Such technologies utilize uniquely quantum phenomena, such as entanglement \cite{entanglement}, coherence \cite{coherence}, and non-Gaussianity \cite{nonGstates}, to achieve capabilities unattainable by conventional means. The systematic characterization of these and other quantum features has given rise to the field of {\it quantum resource theories} \cite{QRTs}, in which quantum phenomena are rigorously analyzed to explore what advantages they may provide. 
%assess their role in enabling quantum advantages.

%The introduction of quantum mechanics to modern technologies has led to the development of quantum technologies that promise to outperform their classical counterparts in many areas. Such technologies utilize explicitly quantum features, such as entanglement, coherence and non-Gaussianity, in order to go beyond classical capabilities. The characterization of such properties has developed the field of quantum resource theories, in which quantum properties are rigorously studied to see what advantages they may provide.

 Quantum resource theories categorize states as either free or resourceful, with the expectation that resource states are advantageous over those that are free. The study then aims to quantify how useful a particular state would be in some relevant task. Beyond specific examples, resource theories can also be studied in a general sense, allowing for the understanding of properties common in a wide range of resource scenarios and unveiling links between different quantum properties.  One such link has been successfully established by focusing on a resource measure known as {\it generalized robustness} \cite{original_robustness}, which turns out to quantify how useful a resource state is in a class of hypothesis testing tasks commonly referred to as {\it channel discrimination}. This result was first shown in finite dimensions for general resource theories with convex sets of free states \cite{robustnessofcoherence,finiteconvexrobustness}, and later extended to generalized probabilistic theories \cite{GPTRobustness}, infinite-dimensional convex theories \cite{infiniteconvex,longinfiniteconvex}, and resource theories with non-convex free sets in finite dimensions \cite{finitenonconvex,longfinitenonconvex}. These results cover several cases of physical relevance, including the convex resource theories of quantum coherence and entanglement \cite{coherence,entanglement} (whose resource content cannot be increased by classical mixing) and the non-convex resource theories of total correlations and quantum correlations beyond entanglement \cite{starresourcetheories,ABC2016} (where shared randomness cannot be regarded as free).  Similar general results have been obtained for channel exclusion tasks by quantifying advantages in terms of resource weight measures \cite{WoR} and also for a family of quantum betting tasks which interpolates between channel discrimination and exclusion \cite{Betting}. 
 All these results are nevertheless not universal, as they leave out the nontrivial case of quantum resources defined by a non-convex set of free states in infinite dimensions, which play a prominent role in continuous variable quantum information processing. %This means in particular that previous literature on quantum resource theories falls short of providing a universal operational interpretation for  The latter  leave out the important case of infinite-dimensional resource theories characterized by a non-convex set of free states. 

 As an alternative to the discrete variable approach, {\it continuous variable quantum information} \cite{QI_with_CV,alessiobook} is growing as a promising area, largely due to its ease of implementation with current quantum optical setups. Although quantum technologies in this regime are rapidly advancing, less is understood about the specific resources needed to underpin the quantum advantages such technologies aim to achieve. This is largely due to the fact that mathematical techniques developed for the discrete variable setting do not necessarily apply to continuous variables, as the infinite-dimensional Hilbert spaces involved are substantially more complex than their finite-dimensional counterparts. Several results in this area make restricting assumptions on the spaces involved, such as a bounded mean energy which effectively makes the Hilbert space finite \cite{continuity_too_strong}, or work solely within the simplified Gaussian regime. 
 
 Gaussian states admit an elegant mathematical formalism \cite{CVQI_and_beyond} and arise naturally in physical systems as the ground and thermal equilibrium states of harmonic Hamiltonians \cite{alessiobook}. They can be created and manipulated with relative ease using linear optics, which makes Gaussian quantum information a popular area to explore theoretically and experimentally \cite{GaussianQI}. However, simplicity comes at the price of significant limitations. While basic protocols in domains such as quantum communication and metrology can be implemented using Gaussian states and Gaussian operations alone, their performance may be improved and optimized only by resorting to non-Gaussian elements \cite{estimation,estimationng,photonsubtraction,teleportationng,gaussianqkd,gaussianqkdlimited,variationalcentrone}. Additionally, a series of no-go results involving e.g.~resource distillation, universal quantum computation and error correction, mean that leaving the Gaussian world is necessary in order to even achieve these and other quantum information tasks \cite{noGdistillation,noGdistillation2,noGdistillation3,gaussianQRTs,noGEC,wignernegativity,bertadista}.

 We see therefore that {\it non-Gaussianity} emerges quite naturally as an essential resource for continuous variable quantum information processing \cite{nonGstates,quantifyingnG,nonGfilip}. This raises the central question: Is  any non-Gaussian state useful for something? Does any form and amount of non-Gaussianity translate into a concrete advantage in some quantum task? In previous literature, a distinction has often been made between non-Gaussian states which cannot be obtained as convex mixtures of Gaussian states --- referred to as possessing `quantum' or `genuine' non-Gaussianity --- whose usefulness has been more thoroughly investigated \cite{convexG1,convexG2,infiniteconvex,longinfiniteconvex,FilipMista,Tommoguitar,stellartois,stellaresource}, and states within the convex hull of the (non-convex) set of Gaussian states, whose resourcefulness has not been generally appreciated, as they admit an efficient classical simulation \cite{wignernegativity} while still being non-Gaussian. Quantitative approaches to witness or measure non-Gaussianity have been proposed e.g.~in terms of how distinguishable a given non-Gaussian state is from a reference Gaussian state (for general non-Gaussianity) \cite{quantifyingnG,gaussianHS,gaussianRE,GREmimimized,husimi,negentropy} or from the convex hull of Gaussian states (for the `genuine' variation) \cite{infiniteconvex,longinfiniteconvex,FilipMista,Tommoguitar,hahn24}.   However, as anticipated, the operational framework for the generalized robustness so far does not include the case of infinite-dimensional, non-convex resource theories, and hence as of yet cannot be used to study non-Gaussianity and the advantages it may give in its most general incarnation. 

In this paper, we show that every resource state in any resource theory, without assuming convexity or finite dimensions, can provide an advantage in some channel discrimination task, even if it lies within the convex hull of free states [Section~\ref{sec:TwoCopies}]. We extend the generalized robustness to the case of infinite-dimensional, non-convex resource theories, prove its faithfulness and monotonicity under general free operations, and give it two operational interpretations: the worst case advantage with respect to a convex set decomposition of the free states in single-copy channel discrimination, and an upper bound on the advantage given by a resource state in multi-copy channel discrimination [Section~\ref{sec:generalroberto}]. This approach also provides observable bounds to estimate the robustness of any resource in experiments by measuring suitable witnesses on one or two copies of the system of interest. 

The framework we establish in this paper can finally be applied to study non-Gaussianity as a resource in full generality [Section~\ref{sec:NG}].  Specializing to this case, we show that the generalized robustness of non-Gaussianity is lower semi-continuous and can be calculated exactly in some relevant cases, including Fock states for any number of photons. We further show that simple mixtures of Gaussian states, which have neither `genuine' non-Gaussianity nor optical nonclassicality, can nonetheless outperform all Gaussian states at some channel discrimination task, and we provide lower bounds on the advantage characterized by the robustness of non-Gaussianity. Our results prove in particular that {\it all} non-Gaussian states, within and outside of the convex hull of Gaussian states, yield provably useful resources for quantum communication and discrimination technologies, thus providing an affirmative answer to the main question raised in this work.

\section{Preliminaries}\label{sec:backup}
Here we begin by introducing the notation and background material necessary for this paper.

Throughout this paper, $\mathcal{H}$ refers to a separable, not necessarily finite-dimensional, Hilbert space; $\bbh$ is the space of bounded operators on $\mathcal{H}$, $\tth\subseteq\bbh$ is the space of trace class operators, $\kkh \subseteq \bbh$ is the space of compact operators, and $\ddh\subseteq\tth$ is the space of density operators, or states, on $\mathcal{H}$. The trace norm and the operator norm are denoted by $\|\cdot\|_1$ and  $\|\cdot\|_{\infty}$, respectively. We further use $\cone(X)$ to denote the cone generated by the set $X\subseteq \tth$, that is $\cone(X):=\{\mu x \colon x\in X, \mu\in\mathbb{R}_{\geq 0}\}$.

\subsection{Topologies in infinite dimensions}
When working with an infinite-dimensional state space, we encounter a rich topological structure, including different topologies that no longer coincide (see e.g. \cite{Megginson}).
%e.g. Corollary 2.6.3 says norm and w* are the same iff finite dimensions
We provide here a brief overview of the topologies we use in this paper.

The trace norm $\|\cdot\|_1$ induces the {\em trace norm topology} on $\tth$, in which a sequence of trace class operators $\{T_n\}_n$ is said to converge to $T$ if $\lim_{n\to\infty}\|T-T_n\|_1=0$, and we write this as $T_n\xrightarrow{\text{tn}} T$. Due to the operational significance of the trace norm in quantifying state distinguishability \cite{Holevo1973,Helstrom1976}, the trace norm topology is an important topology when working with quantum states, and we will use this topology by default unless otherwise specified.

A different norm topology that we will also use in this paper is the {\em Hilbert-Schmidt topology}, induced by the Hilbert-Schmidt norm $\|X\|_2:=\tr[XX^{\dagger}]^{\frac{1}{2}}$. A sequence $\{T_n\}_n$ converges in the Hilbert-Schmidt topology to $T$ if $\lim_{n\to\infty}\|T-T_n\|_2=0$, which we write as $T_n\xrightarrow{\text{H-S}} T$.

Another topology that will prove useful is the {\em weak* topology} on $\tth$. This is the topology induced on $\tth$ by its predual, the space of compact operators $\kkh$. A sequence of trace class operators $\{T_n\}_n$ is said to converge in the weak* topology to $T\in\tth$ if $\lim_{n\to\infty}\tr[T_nK]=\tr[TK]$ for all $K\in \kkh$, and we write this as $T_n\xrightarrow{\text{w*}} T$. The usefulness of this topology in this paper is due to the Banach-Alaoglu theorem \cite[Theorem~2.6.18]{Megginson}, which says the trace norm unit ball in $\tth$ is compact in the weak* topology, and hence this topology lends itself well to arguments involving sequences.

Since we are working in infinite dimensions, all these topologies are different. It is, however, worth noting that when working specifically with states, convergence of a sequence of states to another state is equivalent in all of these topologies, i.e. for a sequence of states $T_n\in\ddh$ and another state $T\in\ddh$, we have (see e.g.~\cite{SWOT} and references therein) 
\begin{equation}
    T_n\xrightarrow{\text{tn}} T \iff T_n\xrightarrow{\text{H-S}} T \iff T_n\xrightarrow{\text{w*}} T.
\end{equation}

\subsection{Quantum resource theories}

\begin{figure}[t]
    \centering
    \subfloat[]{\includegraphics[height=2.5cm]{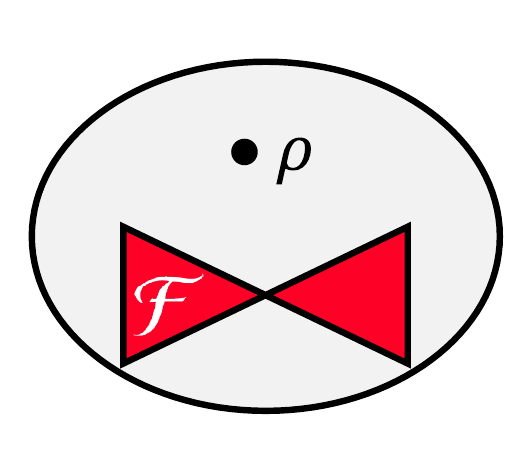}\label{fig:pringlesman}}\hfill \subfloat[]{\includegraphics[height=2.5cm]{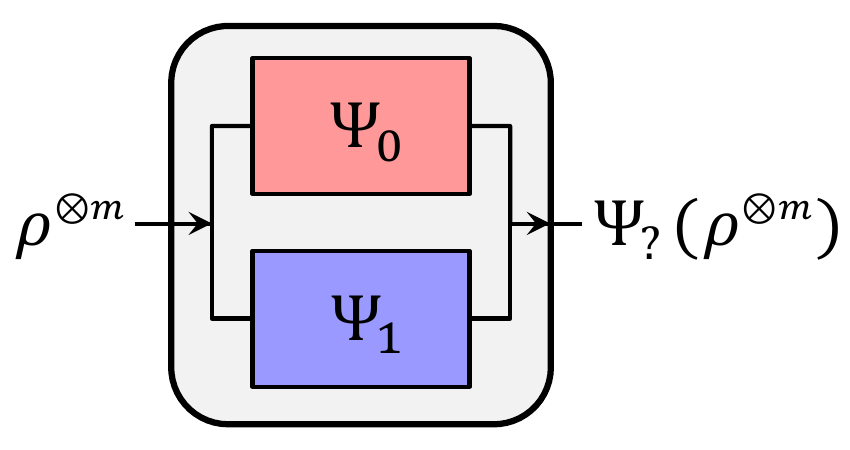}\label{fig:pringleschannel}}
\caption{\label{fig:pringles} (a) Illustration of a general resource theory. The oval represents the space $\ddh$ of states on a finite- or infinite-dimensional Hilbert space $\mathcal{H}$. The bow tie represents the closed, but possibly non-convex set $\ff\subseteq\ddh$ of free states. (b) Schematic of a $m$-copy binary channel discrimination task. The goal is to maximize the probability $p_{succ}$ of correctly guessing which channel was applied to the factorized input state $\rho^{\otimes m}$, out of an ensemble of possible channels $\{\Psi_i\}_i$ with prior probabilities $\{p_i\}_i$, by implementing an optimal measurement $\{M_i\}_i$ at the output. In this paper we show that, in any resource theory, every $\rho \notin \ff$ can strictly outperform all free states $\sigma \in \ff$ at some $m$-copy channel discrimination task (with $m=1$ if $\ff$ is convex and $m \geq 2$ otherwise).}
\end{figure}

Quantum resource theories are motivated by the fact that, in a realistic setting, the accessible states and operations will be restricted \cite{QRTs}. In this framework, the set of all possible quantum states is partitioned into the set of free states $\ff\subseteq\ddh$ (those that are readily available in the considered context) and the set of resource states $\ddh \backslash \ff$. The idea is that a non-free state $\rho\notin\ff$ can then potentially be used as a \emph{resource} to enable or improve the performance of useful tasks. The set of operations that can be implemented under the given restrictions are known as free operations, and have the requirement that $\ff$ is closed under their action; in other words, free operations cannot create resource states out of free states, and do not require any resource expenditure to be implemented. 

Resource theories can be considered in specific cases (entanglement \cite{entanglement} and coherence \cite{coherence} being  prime examples), or in a general setting, allowing for the study of results that apply to broad classes of resources beyond their physical specifications. Here, we work mostly in this general case. We make the standard assumption that $\ff$ is closed (in the trace norm topology) -- if a state can be approximated by a sequence of free states then it must also be free. To allow for full generality in the applicability of our results, this is the only assumption we make [Fig.~\ref{fig:pringlesman}].

Note that although convexity of $\ff$ is often also assumed in a general resource theory \cite{BrandaoGour}, we do not make this assumption here. This is because there exist several interesting and well-motivated sets of free states which are naturally non-convex, a notable example here being the set of Gaussian states of continuous variable systems \cite{CVQI_and_beyond}. Additionally, there are situations in which randomness can be costly and may be accounted for as a resource \cite{randomnessresource}, and therefore it is not necessarily always productive to consider mixtures of free states as free \cite{finitenonconvex,longfinitenonconvex,robustnesscontinuity,starresourcetheories}.

It is a natural assumption that states with a higher resource content should be more useful. Given $\rho\in\ddh\backslash\ff$, the task is then to quantify the amount of resource in $\rho$. This is done by introducing a resource monotone $\mathcal{M}:\rho \mapsto [0,\infty)$. Whilst there is no unique way of defining such an $\mathcal{M}$, two standard conditions may be imposed in order to ensure the monotone is meaningful: faithfulness, $\mathcal{M}(\rho) = 0 \iff \rho \in \ff$, and monotonicity under free operations, $\mathcal{M}(\Phi(\rho))\leq \mathcal{M}(\rho)$ for any free operation $\Phi$. It is also instructive to consider monotones with some {\it operational} meaning, i.e. relate the resource content of a state with respect to that monotone to the relative advantage it gives in some specific task of interest.

When working with monotones in the continuous variable case, continuity issues often arise. Intuitively, continuity of a monotone is desirable since if two states can be related by only a slight perturbation, their resource content should be similar. In the finite-dimensional case, this is often trivially given by the compactness of $\ff$, whereas when defined the same way in infinite dimensions, many quantities are discontinuous everywhere \cite{continuity_too_strong}. It is therefore necessary to impose a weaker version of continuity: lower semi-continuity. A monotone $\mathcal{M}$ is lower semi-continuous if it satisfies the property $\liminf_{n\to\infty}\mathcal{M}(\xi_n)\geq \mathcal{M}(\rho)$ for any sequence of states $\xiseq \rightarrow\rho$. This direction of semi-continuity imposes that approximating a state by a sequence of states does not give a lower value of resource than in the actual state, i.e. taking the limit will not increase the resource.

\subsection{Channel discrimination tasks}

In a channel discrimination task, we are provided with a channel ensemble $\{p_i,\Psi_i\}_i$, from which the channel $\Psi_i$ is picked with probability $p_i$ and applied to an input state $\eta$. The task is to correctly guess which channel was picked by measuring the output, and the probability of success can be maximized by optimizing the choice of $\povm$ (positive operator valued measure) implemented at the output. The success probability of correctly guessing from the ensemble $\{p_i,\Psi_i\}_i$ using input state $\eta$ and using POVM $\{M_i\}_i$ is denoted \cite{robustnessofcoherence,finiteconvexrobustness}
\begin{equation}
p_{succ}(\{p_i,\Psi_i\}_i,\{M_i\}_i,\eta):=\sum_i p_i \tr[\Psi_i(\eta) M_i].    
\end{equation}

These tasks are very relevant in quantum information, since any scenario in which we want to discover which of a set of processes are occurring can be reformulated as a channel discrimination task. For example, promising uses of binary discrimination tasks include quantum illumination \cite{IlluminationIntro,GaussianIllumination,ExperimantalIllumination} or quantum reading \cite{ReadingIntro,ReadingAndIllumination,ExperimentalReading}. 

We can further consider the more general case of multi-copy channel discrimination tasks [Fig.~\ref{fig:pringleschannel}], in which the channels $\{\Psi_i\}_i$ act on, say, $m$ copies of an input state $\eta$ \cite{finitenonconvex,longfinitenonconvex}. The probability of success is accordingly denoted as $p_{succ}(\{p_i,\Psi_i\}_i,\{M_i\}_i,\eta^{\otimes m})$.
In general, we say that a state $\rho$ provides an {\it advantage} over another state $\tau$ in a given $m$-copy channel discrimination task $\{p_i,\Psi_i\}_i$, if
\begin{equation}
\begin{aligned}
&\sup_{\{M_i\}_i}p_{succ}(\{p_i,\Psi_i\}_i,\{M_i\}_i,\rho^{\otimes m}) \\  &> \sup_{\{M_i\}_i}p_{succ}(\{p_i,\Psi_i\}_i,\{M_i\}_i,\tau^{\otimes m}).
\end{aligned}
\end{equation}
A nontrivial question is thus whether {\it any} quantum resource will give an advantage over free states in some instance of channel discrimination. In the single-copy scenario ($m=1$), the answer is affirmative provided the set $\ff$ of free states is convex \cite{finiteconvexrobustness}. Extending this result to arbitrary resource theories in the multi-copy scenario is the focus of Section~\ref{sec:TwoCopies}.

\subsection{Continuous variables}\label{sec:CV}
A continuous variable bosonic system is described as a system of $m$ modes, each corresponding to a quantum harmonic oscillator with a different frequency \cite{alessiobook}. In the phase space representation, each mode $i$ has an associated pair of quadrature operators $\hat{x}_i, \hat{y}_i$. The overall system is then described by the collection of these: $\mathbf{\hat{r}}=(\hat{x}_1,\hat{y}_1,...,\hat{x}_m,\hat{y}_m)$. These satisfy the commutation relations $[\hat{r}_j,\hat{r}_k]=i\Omega_{jk}$, where $\boldsymbol{\Omega}$ is the $m$-mode symplectic form, given by 
$$\boldsymbol{\Omega} = \bigoplus_{l=1}^m \begin{pmatrix}
0 & 1\\
-1 & 0
\end{pmatrix}.$$
The eigenvalue associated with eigenvector $\ket{r_i}$ of the quadrature operator $\hat{r}_i$ is denoted $r_i$. 
Information is then encoded in the continuous spectra of these operators.

Continuous variable states can be represented in phase space by their (symmetrically ordered) characteristic function $\chi^{\rho}(\mathbf{r})$, defined as
\begin{equation}
    \chi^{\rho}(\mathbf{r}) = \tr[{\rm e}^{-i\mathbf{r}^\top \boldsymbol{\Omega} \mathbf{\hat{r}}}\rho].
\end{equation}

Alternatively, states can be described by quasi-probability distributions. One such distribution is the Wigner function, given by the Fourier transform of $\chi^{\rho}(\mathbf{r})$, and defined for $m$ modes as 
\begin{equation}
   \mathcal{W}^\rho(\mathbf{r}) = \frac{1}{\pi^{2m}} \int_{\mathbb{R}^{2m}} \chi^\rho(\mathbf{s}) {\rm e}^{i\mathbf{s}^\top\boldsymbol{\Omega} \mathbf{r}} d^m \mathbf{s}.
\end{equation}
%\frac{1}{\pi^m}\int_{\mathbb{R}^m}\bra{\mathbf{x}+\mathbf{s}}\rho\ket{\mathbf{x}-\mathbf{s}}{\rm e}^{2i\mathbf{sy}}d^m\mathbf{s}
%This resembles a probability distribution in the sense that it is normalized and the marginal distributions give the correct probability distributions for each variable. However, the quantum nature of some states is reflected in the fact that $W_\rho(\mathbf{q},\mathbf{p})$ may be negative.

Within the continuous variable regime, the set of Gaussian states plays a special role \cite{CVQI_and_beyond}. Gaussian states are defined to be all states with a Gaussian Wigner function 
 \begin{equation}\label{gaussianwigner}
 \mathcal{W}^\rho(\mathbf{r}) = \frac{1}{\pi^m\sqrt{\det(\boldsymbol{V})}}{\rm e}^{{-(\mathbf{r}-\boldsymbol{\mu})^\top \boldsymbol{V}^{-1}(\mathbf{r}-\boldsymbol{\mu})}}
 \end{equation}
 where $\rho$ is uniquely specified by two quantities: the vector of first moments, $\boldsymbol{\mu}:=\tr[\rho \mathbf{\hat{r}}]$, and the covariance matrix, $(\boldsymbol{V})_{jk}:=\tr[\rho\{\hat{r}_j-\mu_j,\hat{r}_k-\mu_k\}]$, satisfying the {\it bona fide} condition $\boldsymbol{V}+i\boldsymbol{\Omega}\geq \boldsymbol{0}$ \cite{alessiobook}.

\section{Operational advantage of general resources in two-copy channel discrimination}\label{sec:TwoCopies}
In this Section, we show that every resource state in any resource theory provides an advantage in some two-copy channel discrimination task, see Fig.~\ref{fig:pringles}. The proof strategy involves directly constructing such a task, using the properties of {\it resource witness} operators. 

When $\ff$ is convex, the hyperplane separation theorem guarantees the existence of a linear (i.e.~acting on a single copy of a state) witness operator \cite{witnessexistence}. 
When the set $\ff$ is non-convex, such a linear operator is no longer guaranteed to exist. We therefore turn to multi-copy witness operators. These were first introduced in entanglement theory \cite{swaptrick}, and amount to a witness operator $W\in\mathcal{B(H}^{\otimes m})$ that acts on multiple copies of a state $\rho$ in order to separate it from $\ff$, i.e. 
\begin{equation}\label{witnessdef}
\begin{array}{l}
\tr[W\rho^{\otimes m}]<0\,, \\
\tr[W\sigma^{\otimes m}]\geq 0 \ \ \forall \sigma\in\ff.
\end{array}
\end{equation}
In Ref.~\cite{finitenonconvex}, a family of multi-copy witnesses based on the generalized Bloch description of finite-dimensional quantum states was used to show an advantage in channel discrimination even for non-convex sets of free states; however, such a construction cannot be straightforwardly extended to continuous variable systems. 

Here we present  a general construction to show that nonetheless multi-copy witnesses always exist in a general finite- or infinite-dimensional resource theory.

\begin{proposition} \label{WitnessExistence}
    For any set of free states $\ff\subseteq\ddh$ that is closed in the trace norm topology and some resource state $\rho\in\ddh \backslash \ff$, there exists a two-copy witness operator $W$ such that $\tr[W\rho^{\otimes 2}]<0$ and $\tr[W\sigma^{\otimes 2}]\geq 0$ for all $\sigma \in \ff$.
\end{proposition}
\begin{proof}
    Consider the two-copy  operator given by 
    \begin{equation} \label{lockness}
    W(\rho, \varepsilon) = \swap+(\tr[\rho^2]-\varepsilon)\mathds{1}\otimes \mathds{1}-2\rho \otimes \mathds{1}
    \end{equation}
where the $\swap$ operator is defined via $\swap \ket{\phi}\otimes \ket{\psi} = \ket{\psi}\otimes \ket{\phi}$ and has the property $\tr[\swap\ \eta^{\otimes 2}] = \tr[\eta^2]$ \cite{swaptrick}. When evaluated on two copies of a state $\eta$, this gives \cite{finitenonconvex}
\begin{equation}\tr[W\eta^{\otimes 2}] = \tr[(\rho-\eta)^2]-\varepsilon
\end{equation}
and is therefore a witness provided $\varepsilon$ is picked such that $0<\varepsilon \leq \tr[(\rho -\sigma)^2]$ for all $\sigma \in \ff$. 

We now need to argue that such an $\varepsilon$ can always be chosen. To this end, assume towards a contradiction that no such $\varepsilon$ exists, that is, $\tr[(\rho -\sigma)^2]$ can be made arbitrarily small. This means there exists some sequence $\{\sigma_n\}_n$ in $\ff$ such that $\sigma_n \xrightarrow{\text{H-S}} \rho$, and hence also implies $\sigma_n \xrightarrow{\text{tn}} \rho$. This contradicts the assumption that $\ff$ is closed in the trace norm topology, and therefore such an $\varepsilon$ must always exist.
\end{proof}

\begin{remark}
    The two-copy witness (\ref{lockness})  is based on the distance $\tr[(\rho-\eta)^2]$. The construction in Proposition~\ref{WitnessExistence} can be generalized to define a family of multi-copy witnesses based on the distance $\tr[(\rho-\eta)^{2n}]$ for any integer $n\geq 1$.
\end{remark}
    
    Explicitly, the aim is to construct a Hermitian operator $W_{2n}(\rho,\varepsilon)$ such that 
    \begin{equation}
    \begin{aligned}
    &\tr\left[W_{2n}\eta^{\otimes 2n}\right]\\
    &=\tr\left[(\rho-\eta)^{2n}\right]-\varepsilon\\
    &=(\tr[\rho^{2n}]-\varepsilon)+\tr\left[\sum_{k=1}^{2n}(-1)^k\!\!\!\!\sum_{permutations}\!\!\!\!\rho^{2n-k}\eta^k\right].
    \end{aligned}
    \end{equation}

        The first term in the final equation is achieved by the operator $(\tr[\rho^{2n}]-\varepsilon)\mathds{1}^{\otimes 2n}$.
    To construct the remaining terms, we use the $m$-party unitary shift operator $V_m\in \mathcal{B(H}^{\otimes m})$, which has the property $\tr[V_m (A_1\otimes A_2\otimes \ldots\otimes A_m)]=\tr[A_1A_2\ldots A_m]$ for $A_1,\ldots,A_m\in\tth$ \cite{Isham_1994}. Notice that $V_2 \equiv \swap$, while $V_m \neq V_m^{\dagger}$ for $m >2$. A term of the form  $\tr[\rho^{a_1}\eta^{b_1}\rho^{a_2}\eta^{b_2}\ldots\rho^{a_k}\eta^{b_k}]$ can then be achieved using an operator of the form 
    \begin{equation}
        \begin{aligned}
        \Big(V_{(\sum_{i=1}^kb_i)}\rho^{a_1}\otimes\mathds{1}^{\otimes b_1-1}\otimes \rho^{a_2}\otimes \mathds{1}^{\otimes b_2-1}\otimes\ldots\\
        \ldots\otimes\rho^{a_k}\otimes\mathds{1}^{\otimes b_k-1}\Big)\otimes\mathds{1}^{\otimes (2n-\sum_{i=1}^kb_i)}
        \end{aligned}
   \end{equation}
    where we take $V_1:=\mathds{1}$. An alternating sum of these operators will complete the construction of a quasi-witness $\tilde{W}_{2n}$ (referring to the fact that $\tilde{W}_{2n}$ identifies a resource but is not itself Hermitian). Since $\tr[\tilde{W}_{2n}\eta^{\otimes 2n}]=\tr[(\rho-\eta)^{2n}]-\varepsilon$ is real valued, we note $\tr[\tilde{W}_{2n}\eta^{\otimes 2n}]=\tr[\frac12(\tilde{W}_{2n}+\tilde{W}_{2n}^\dagger)\eta^{\otimes 2n}]$ \cite{swaptrick}, and thus a Hermitian $2n-$copy witness can be constructed as 
    \begin{equation}
    W_{2n}=\frac12(\tilde{W}_{2n}+\tilde{W}_{2n}^\dagger).
    \end{equation}
    For example, a general $4-$copy witness is given by 
    \begin{equation}
W_4(\rho,\varepsilon)=\frac12(\tilde{W}_4+\tilde{W}_4^\dagger),
    \end{equation}
    where
    \begin{equation}
\begin{aligned}
\tilde{W}_4:=&(\tr[\rho^4]-\varepsilon)\mathds{1}^{\otimes 4}-4\rho^3\otimes\mathds{1}^{\otimes3}\\&+4(V_2\otimes\mathds{1}^{\otimes 2})(\rho^2\otimes\mathds{1}^{\otimes3})
+2V_2\rho^{\otimes 2}\otimes\mathds{1}^{\otimes 2}\\&-4(V_3\otimes\mathds{1})(\rho\otimes\mathds{1}^{\otimes3})+V_4,
\end{aligned}
\end{equation}
and is such that $\tr[W_4\eta^{\otimes 4}]=\tr[(\rho-\eta)^4]-\varepsilon$. The observable $W_4(\rho,\varepsilon)$ is therefore a resource witness provided we pick $\varepsilon$ such that $0<\varepsilon\leq\tr[(\rho-\eta)^4]$: this can always be done, even in infinite dimensions, via an analogous argument to that used in Proposition~\ref{WitnessExistence}.

Our analysis so far in this Section shows that, although there may not exist a linear witness operator in some resource theory, there will always exist a multi-copy witness operator. Interestingly, it also shows we only require two copies of the states involved to be able to strictly separate $\rho$ from $\ff$, regardless of the system's dimension and including infinite-dimensional theories.

We can now use this result to show that, for any resource state $\rho$, there always exists a multi-copy channel discrimination task for which $\rho$ performs strictly better than any free state.

\begin{theorem} \label{DiscriminationAdvantage}
     $\rho \in \ddh \backslash \ff$ if and only if there exist two channels $\Psi_0, \Psi_1$ acting on $\mathcal{D(H}^{\otimes 2})$ such that 
    \begin{equation*}\frac{\sup_{\{M_i\}_i}p_{succ}(\{p_i,\Psi_i\}_{i=0,1},\{M_i\}_i,\rho^{\otimes 2})}{\sup_{\sigma\in\ff}\sup_{\{M_i\}_i}p_{succ}(\{p_i,\Psi_i\}_{i=0,1},\{M_i\}_i,\sigma^{\otimes 2})}>1.\end{equation*}
\end{theorem}

\begin{proof}

The ``if" direction follows directly from the assumption that $\ff$ is closed. For the ``only if" direction, we directly construct a channel discrimination task for which $\rho\in \ddh \backslash \ff$ gives an advantage. Given $\rho$, take any bounded, two-copy witness operator (the existence of which is shown in Proposition \ref{WitnessExistence}) and define $X:=\mathds{1}^{\otimes 2} - \frac{W}{\|W\|_{\infty}}$. Note $X$ has the properties $\tr[X\rho^{\otimes 2}]>1$, $0\leq \tr[X\sigma^{\otimes 2}]\leq 1$ for all $\sigma\in\ff$.
Now consider the two quantum channels given by 
    \begin{equation*}
    \begin{aligned}
    &\Psi_0(\eta^{\otimes 2}):= \mbox{$\left(\frac{1}{2} + \frac{\tr[X\eta^{\otimes 2}]}{2\|X\|_{\infty}}\right)\ket{0}\!\bra{0} +\left(\frac{1}{2} - \frac{\tr[X\eta^{\otimes 2}]}{2\|X\|_{\infty}}\right) \ket{1}\!\bra{1}$}\\
    &\Psi_1(\eta^{\otimes 2}):=  \mbox{$\left(\frac{1}{2} - \frac{\tr[X\eta^{\otimes 2}]}{2\|X\|_{\infty}}\right)\ket{0}\!\bra{0} + \left(\frac{1}{2} + \frac{\tr[X\eta^{\otimes 2}]}{2\|X\|_{\infty}}\right)\ket{1}\!\bra{1}$}
    \end{aligned}
    \end{equation*}
    where $\ket{0},\ket{1}$ are mutually orthogonal states.
    Consider the task of determining which of these two channels was applied when there was an equal probability of picking either channel. By the Holevo-Helstrom theorem \cite{Holevo1973,Helstrom1976}, the maximum success probability of correctly guessing is 
    \begin{equation}
    \begin{aligned}
        \sup_{\{M_i\}_i}\ &p_{succ}\left(\{\tfrac{1}{2},\Psi_i\}_i,\{M_i\}_i, \eta^{\otimes 2}\right) \\ &= \frac{1}{2}\left(1+\frac{1}{2} \|(\Psi_0 - \Psi_1)\eta^{\otimes  2}\|_1\right)\\
        &=\frac{1}{2}\left(1+\frac{\tr[X\eta^{\otimes 2}]}{\|X\|_{\infty}}\right)\\
        & \begin{cases}
            > \frac{1}{2}\left(1+\frac{1}{\|X\|_{\infty}}\right) & \text{for} \ \eta = \rho\\
            \leq \frac{1}{2}\left(1+\frac{1}{\|X\|_{\infty}}\right) & \text{for} \ \eta = \sigma \in \ff
        \end{cases}
    \end{aligned}
    \end{equation}
    and therefore the ratio of success probabilities is
    \begin{equation*}
    \begin{aligned}
       & \frac{\sup_{\{M_i\}_i}p_{succ}(\{p_i,\Psi_i\}_{i=0,1},\{M_i\}_i,\rho^{\otimes 2})}{\sup_{\sigma\in\ff}\sup_{\{M_i\}_i}p_{succ}(\{p_i,\Psi_i\}_{i=0,1},\{M_i\}_i,\sigma^{\otimes 2})} \\ & \qquad >\frac{\frac{1}{2}\left(1+\frac{1}{\|X\|_{\infty}}\right)}{\frac{1}{2}\left(1+\frac{1}{\|X\|_{\infty}}\right)} = 1.
    \end{aligned}\vspace*{-0.4cm}
    \end{equation*} 
\end{proof}

This shows an operational advantage of every resource state in any resource theory in the channel discrimination setting. Remarkably, we only need two copies of the states to unlock such an advantage, even in the most general scenario. Since the only assumption we make is that $\ff$ is closed, this result applies to every resource theory, generalizing and complementing  the finite-dimensional, non-convex case \cite{finitenonconvex,longfinitenonconvex} and the infinite-dimensional, convex case \cite{infiniteconvex,longinfiniteconvex}.

\section{Generalized robustness monotones in non-convex, continuous variable resource theories}\label{sec:generalroberto}

We begin this Section by introducing the generalized robustness monotones considered in this paper and discussing their properties.

\begin{definition}[Generalized robustness]\label{robDef} The generalized robustness of a state $\rho$ with respect to the set of free states $\ff$ is given by 
\begin{equation}
    \rob :=\inf_{\tau\in\ddh}\left\{\lambda\in \mathbb{R}_{\geq 0} \colon \frac{\rho + \lambda \tau}{1+\lambda} = \sigma \in \ff \right\}.
\end{equation}
\end{definition}

The generalized robustness monotone was first introduced in entanglement theory \cite{original_robustness}, and later extended to arbitrary resource theories in finite dimensions \cite{finiteconvexrobustness} and convex resource theories within the framework of general probabilistic theories \cite{GPTRobustness}. The generalized robustness $\rob$ quantifies the amount of resource in a state $\rho$ by the amount of noise that can be added via mixing with another state before all the resource is lost.

\begin{remark}The generalized robustness can equivalently be expressed as
\begin{equation}\label{robDmax}
\begin{aligned} 
    \rob &= \inf_{\sigma \in \ff} \left\{\lambda\in \mathbb{R}_{\geq 0} : \rho \leq (1+\lambda) \sigma \right\} \\
    & = \inf_{\sigma \in \ff} \left\{\exp\big[D_{\max}(\rho \| \sigma)\big] -1 \right\}\,,
\end{aligned}
\end{equation}
where
\begin{equation}\label{Dmax}
    D_{\max}(\rho \| \sigma) := \inf \{\gamma \in \mathbb{R} : \rho \leq {\rm e}^\gamma\ \sigma\},
\end{equation} is the {\it max-relative entropy} between $\rho$ and $\sigma$ \cite{min-maxRE,datta2009max}.  
\end{remark}

As an operational interpretation, when used in discrete variable theories, $\rob$ quantifies the maximal potential advantage given by a resource state $\rho$ in single-copy channel discrimination tasks, provided $\ff$ is convex \cite{finiteconvexrobustness}. In the case where $\ff$ is non-convex, it has been shown to quantify the worst case maximal advantage in single-copy channel discrimination tasks, with respect to a decomposition of $\ff$ into convex subsets \cite{finitenonconvex}.

When used in continuous variable theories, $\rob$ as in Definition \ref{robDef} is no longer guaranteed to be lower semi-continuous. We therefore resort to a modified definition that characterizes the generalized robustness in such a way that it is automatically lower semi-continuous:

\begin{definition}[Lower semi-continuous generalized robustness]\label{lscrDef}
The lower semi-continuous generalized robustness of a state $\rho$ with respect to $\ff$ is given by \cite{infiniteconvex,note1inf}
 \begin{equation}\label{lscr}
 \begin{aligned}
\!\!\!\!\lscr &:= \liminf_{\xi\rightarrow\rho}\mathcal{R_F}(\xi)\\
 &\:= \inf_{\{\tau_n\}_n, \xiseq}\left\{\lambda \in \mathbb{R}_{\geq 0} \colon \frac{\xi_n + \lambda \tau_n}{1+\lambda} = \sigma_n \in \ff, \right. \\
 &\qquad \qquad\!\! \quad\left.\tau_n, \xi_n \in \ddh, \  \xiseq \xrightarrow{\textup{tn}} \rho \right\}.
 \end{aligned}
\end{equation}
\end{definition}
Here we optimize over all possible sequences of states $\xiseq$ such that $\xi_n\xrightarrow{\text{tn}}\rho$.

So far, this lower semi-continuous robustness $\lscr$ has only been studied for convex $\ff$ \cite{infiniteconvex,longfinitenonconvex}. Whilst it is no longer a convex function when this assumption is removed, we can show that $\lscr$ is still monotonic and faithful.

\begin{proposition}\label{propoprops}
  The resource quantifier  $\lscr$  defined with respect to a (not necessarily convex) set $\ff$ is both {\em (a)} monotonic under free operations and {\em (b)} faithful.
\end{proposition}

\begin{proof}
We prove each claim individually.

 To establish (a), from the definition of $\lscr$ we have $\forall \varepsilon > 0$, there exist sequences $\{\sigma(\varepsilon)_{n}\}_n$ in $\ff$ and $\{\xi(\varepsilon)_n\}_n \rightarrow\rho$ such that $\xi_n \leq (1+\lscr + \varepsilon)\sigma_n$ for each $n$. For any completely positive, linear, trace non-increasing free operation $\Phi$, positivity  implies $\Phi(\xi_n)\leq (1+\lscr + \varepsilon)\Phi(\sigma_n)$, and contractivity of the trace norm under $\Phi$ means that $\Phi(\xi_n)\rightarrow \Phi(\rho)$ also. Hence $\lscr$ is a suboptimal value for $\underline{\mathcal{R_F}}(\Phi(\rho))$, i.e. $\underline{\mathcal{R_F}}(\Phi(\rho))\leq \lscr$.
    
 To see (b) holds, if $\rho \in \ff$, clearly  $\lscr = 0$. For the converse, $\lscr=0$ implies that there exists a sequence $\lambda^{(k)} \to 0$ such that $\xi_n^{(k)}-\sigma_n^{(k)} = \lambda^{(k)}(\sigma_n^{(k)} - \tau_n^{(k)})$. For fixed $k$, we have $\|\xi_n^{(k)}-\sigma_n^{(k)}\|_1\leq 2\lambda^{(k)}$, and thus $\|\rho-\sigma_n^{(k)}\|_1 \leq \|\rho - \xi_n^{(k)}\|_1 + \|\xi_n^{(k)}-\sigma_n^{(k)}\|_1 \leq \|\rho - \xi_n^{(k)}\|_1 + 2\lambda^{(k)}$. Taking the limits $n\to \infty$ and $k\to \infty$ we see  that $\sigma_n^{(k)}$ converges to $\rho$, hence $\rho\in\ff$.
\end{proof}

This means $\lscr$ remains a valid resource quantifier in any resource theory, even in the continuous variable case when $\ff$ is non-convex.

A natural question is then to ask when Definitions \ref{robDef} and \ref{lscrDef} are equal. This has already been shown in Ref.~\cite[Proposition~7]{longinfiniteconvex} under the condition that $\ff$ is compact, and therefore the two definitions coincide in finite-dimensional resource theories. With the assumption that $\ff$ is convex, they are also shown to be equal if $\cone(\ff)$ is weak* closed \cite[Theorem~12]{longinfiniteconvex}. Our next result provides a proof that this equivalence in fact holds even in the non-convex case.

\begin{proposition}\label{ClosedCone}
    If $\cone({\ff})$ is closed in the weak* topology, then $\lscr = \rob$.
\end{proposition}
\begin{proof}
    We have $\lscr \leq \rob$ by definition, so it remains to show that $\lscr \geq \rob$.

Let $B_1 := \{X\in \mathcal{T(H)}: \|X\|_1 \leq 1\}$ be the trace norm unit ball and note that via the Banach-Alaoglu theorem \cite[Theorem~2.6.18]{Megginson}, $B_1$ is weak* compact. Consider the space $$\tilde{\mathcal{F}}=\cone({\mathcal{F}})\cap B_1.$$  Since $\cone(\ff)$ is weak* closed by assumption, $\tilde{\ff}$ is weak* (sequentially
\cite{sequential}) compact.

By definition of $\rob$, $\forall \varepsilon > 0$ there exists some $\sigma_{\varepsilon} \in \ff$ with $\rho \leq (1+\rob +\varepsilon)\sigma_{\varepsilon}$. Now apply this to some sequence of states $\xiseq$ such that $\xi_n \xrightarrow{\textup{tn}} \rho$ and pick $\varepsilon = \frac{1}{n}$, which shows there exists a sequence $\{\sigma_n\}_n$ in $\ff$ such that $(1+\mathcal{R_F}(\xi_n)+\frac{1}{n})\sigma_n-\xi_n \geq 0$ for each $n$. 
%Since $\{\sigma_n\}_n$ is in $\tilde{\ff}$, there exists a weak* convergent subsequence $\sigma_{n_k}\xrightarrow{w*} \tilde{\sigma} \in \tilde{\ff}$, where $\tilde{\sigma}$ is of the form $\tilde{\sigma}=\mu \sigma$ for some $0\leq \mu \leq 1$ and $\sigma \in \ff$. 

We now wish to show that $\liminf_{n\to\infty}\mathcal{R_F}(\xi_n)=:\tilde{R}\geq \rob$.
Take some compact operator $K\in\kkh$, such that $K\geq 0$. We then have $$\tr\left[K\left(\left(1+\mathcal{R_F}(\xi_n)+\tfrac{1}{n}\right)\sigma_n-\xi_n\right)\right]\geq 0.$$ 
There exists a sub-sequence $\xi_{n_l}$ of $\xi_n$ such that $\lim_{n_l \to \infty}\mathcal{R_F}(\xi_{n_l})=\tilde{R}$.
Since ${\rm cone}(\tilde{\ff})$ is weak* compact, there also exists a sub-sequence $\sigma_{n_j}$ of $\sigma_n$ that converges to $\mu\sigma \in \tilde{\ff}$ in the weak* topology, for $\sigma\in \mathcal{F}$ and $0\leq \mu\leq 1$.  Taking the overlap of the $n_l$ and $n_j$ terms to give a new sequence indexed by $n_k$ implies
\begin{equation*}
\begin{aligned}
    \lim_{n_k\to\infty}&\tr\left[K\left(\left(1+\mathcal{R_F}(\xi_{n_k})+\tfrac{1}{n_k}\right)\sigma_{n_k}-\xi_{n_k}\right)\right] \\ =& \tr\left[K\left(\left(1+\tilde{R}\right)\mu\sigma-\rho\right)\right]\geq 0.
\end{aligned}
\end{equation*}
Pick now the case $K=\ket{\psi}\!\bra{\psi}$. In this case, positivity of the above equation implies 
$$ \left(1+\tilde{R}\right)\mu\sigma-\rho\geq 0.$$ 
From this we conclude
\begin{equation*}
\begin{aligned}
\rho \leq \left(1+\tilde{R}\right)\mu\sigma &\leq \left(1+\tilde{R}\right)\sigma \\
%&\leq \left(1+\tilde{R}+\varepsilon\right)\sigma.
\end{aligned}
\end{equation*}
And it follows that for any $\xiseq \xrightarrow{tn}\rho$ we have
\begin{equation*}
1+\rob \leq 1+\tilde{R}.
\end{equation*}
Additionally optimizing over all such $\xiseq$ yields $\rob \leq \lscr$, as required.
\end{proof}

Although the requirement in Proposition \ref{ClosedCone} may seem rather abstract, it has been shown to hold true for the free states in many cases of resource theories, most notably separable states \cite[Lemma~25]{longinfiniteconvex} and incoherent states \cite[Lemma~31]{longinfiniteconvex}. Crucially, it also holds for the set of Gaussian states of continuous variable systems \cite[Lemma~34]{attainability}, allowing for an easier evaluation of advantages in one of the most important resource theories in the non-convex, non-finite setting: non-Gaussianity. This will be our focus in Section~\ref{sec:NG}.

\subsection{Interpretations for generalized robustness}\label{sec:interpret}
In this Section, we establish operational interpretations for $\lscr$ when used in a general setting, without convexity or finite-dimensional restrictions. We find it can be used to provide an upper bound on the advantage in multi-copy channel discrimination, and thus remains a valuable quantifier in channel discrimination tasks. We also investigate cases in which $\lscr$ exactly quantifies the worst case maximal advantage in a convex set decomposition of $\ff$, extending results found in Ref. \cite{finitenonconvex} to the continuous variable setting.

\subsubsection{Upper bound in multi-copy channel discrimination}

We begin by defining a multi-copy variant of the generalized robustness monotone.
\begin{definition} The $m$-copy generalized robustness of a state $\rho$  with respect to a free set $\ff$ is defined as
\begin{equation}
    \robm:=\inf_{\tau\in\mathcal{D(H}^{\otimes m})}\left\{\lambda \!\in\! \mathbb{R}_{\geq 0} \colon \frac{\rho^{\otimes m}+\lambda \tau}{1+\lambda}=\sigma^{\otimes m}, \ \sigma \!\in\! \ff \right\}.
\end{equation}
\end{definition}
This quantifies the resource in a state $\rho$ by how much $m$ copies of $\rho$ can tolerate mixture with another state $\tau$ before turning into $m$ copies of some free state $\sigma$. Note, due to requiring $m$ uncorrelated copies of $\sigma \in \ff$ rather than an arbitrary free state in the composite set $\ff^{\otimes m}$, $\robm$ is generally not the same as $\mathcal{R}_{\ff^{\otimes m}}(\rho^{\otimes m})$.

We further define a lower semi-continuous version of $\robm$.

\begin{definition} The lower semi-continuous $m$-copy generalized robustness with respect to $\ff$ is defined as
\begin{equation} 
\begin{aligned}
&\lscrm:=\liminf_{\xi\rightarrow\rho}\mathcal{R}_{\ff}^{(m)}(\xi)\\
&\!\!=\!\!\inf_{\xiseq \to \rho}\inf_{\tau_n\in\mathcal{D(H}^{\otimes m})}\left\{\lambda \in \mathbb{R}_{\geq 0} \colon \frac{\xi_{n}^{\otimes m}+\lambda \tau_n}{1+\lambda}=\sigma_{n}^{\otimes m}, \ \sigma_n \!\in\! \ff \right\}.
\end{aligned}
\end{equation}
\end{definition}

We can relate this multi-copy version of robustness to the single-copy version by using the link (\ref{robDmax}) between the generalized robustness and the max-relative entropy (\ref{Dmax}) \cite{min-maxRE,datta2009max}, along with the known additivity property of the latter under tensor products.

\begin{proposition}\label{prop1m1m}
    For any integer m, we have \begin{equation}\label{oneversusmany}1+\lscrm=\left(1+\lscr\right)^m.\end{equation}
\end{proposition}
\begin{proof}
    \begin{equation}
        \begin{aligned}
           \!\!\!\!  1+\lscrm&=\liminf_{\xi\to\rho}\ \big(1+\mathcal{R}_\mathcal{F}^{(m)}(\xi)\big)\\
             &=\liminf_{\xi\to\rho} \inf_{\sigma\in\ff}\exp \big(D_{\max}(\xi^{\otimes m}\|\sigma^{\otimes m})\big)\\
             &=\liminf_{\xi\to\rho} \inf_{\sigma\in\ff}\exp \big(m D_{\max}(\xi\|\sigma)\big)\\
             &=\liminf_{\xi\to\rho}\ \big(1+\mathcal{R_F}(\xi)\big)^m\\
             &=\big(1+\lscr\big)^m.
        \end{aligned}
    \end{equation}
    \end{proof}
Note that the faithfulness and monotonicity under free operations of $\lscrm$ follow directly from this relation.

Starting by adapting part of the proof in Ref. \cite[Theorem~10]{longinfiniteconvex}, we can now show the multi-copy lower semi-continuous robustness $\lscrm$ upper bounds the maximal advantage given by $\rho$ in multi-copy channel discrimination tasks. The link between $\lscrm$ and $\lscr$ then allows us to show that $\lscr$ can directly be used to upper bound a regularization of such an advantage. 

\begin{theorem}\label{theoremupperbody}
In a general resource theory, the lower semi-continuous generalized robustness of a state $\rho$ upper bounds the regularized advantage enabled by $\rho$ over any free state in $m$-copy channel discrimination, 
    \begin{equation*}
    \begin{gathered}
    \left(\sup_{\{p_i,\Psi_i\}_i}\frac{\sup_{\{M_i\}_i}p_{succ}(\{p_i,\Psi_i\}_i,\{M_i\}_i,\rho^{\otimes m})}{\sup_{\{M_i\}_i}\sup_{\sigma\in\ff}p_{succ}(\{p_i,\Psi_i\}_i,\{M_i\}_i,\sigma^{\otimes m})}\right)^{\frac1m} \\   \leq 1+\lscr.
    \end{gathered}
    \end{equation*}
\end{theorem}

\begin{proof}
    Consider a sequence of the form $\{\xi_n^{\otimes m}\}_n = \{(1+\lambda)\sigma_n^{\otimes m} - \lambda \tau_n\}_n$ such that $\xiseq \to \rho$. For a channel ensemble $\{p_i,\Psi_i\}_i$ and POVM $\{M_i\}_i$, the success probability of correctly guessing with $m$ copies of a state is 
    \begin{equation*}
        \begin{aligned}
&p_{succ}(\{p_i,\Psi_i\}_i,\{M_i\}_i,\rho^{\otimes m}) \\
    &= \sum_ip_i\tr[M_i\Psi_i(\rho^{\otimes m})]\\ 
    &=\lim_{n\to\infty}\sum_ip_i\tr[M_i\Psi_i(\xi_n^{\otimes m})] \\
    %&= \lim_{n\to \infty} \sum_ip_i\tr[M_i\Psi_i((1+\lambda)\sigma_n^{\otimes m}-\lambda\tau_n)]\\
    &\leq \limsup_{n\to\infty}\sum_ip_i\tr[M_i\Psi_i((1+\lambda)\sigma_n^{\otimes m})]\\
    &\leq (1+\lambda)\sup_{\tilde{\sigma}\in\ff}\sum_ip_i\tr[M_i\Psi_i(\tilde{\sigma}^{\otimes m})]\\
    &=(1+\lambda)\sup_{\tilde{\sigma}\in\ff}p_{succ}(\{p_i,\Psi_i\}_i,\{M_i\}_i,\tilde{\sigma}^{\otimes m}).
        \end{aligned}
    \end{equation*}
    Since this holds for any $\lambda$ such that $\frac{\xi_{n}^{\otimes m}+\lambda \tau_n}{1+\lambda}=\sigma_{n}^{\otimes m}$ for some $\sigma_n \in \ff$, we have
    \begin{equation*}
        \begin{aligned}
    &p_{succ}(\{p_i,\Psi_i\}_i,\{M_i\}_i,\rho^{\otimes m})\\ &\leq \inf\left\{(1+\lambda)\sup_{\sigma\in\ff}p_{succ}(\{p_i,\Psi_i\}_i,\{M_i\}_i,\sigma^{\otimes m}) \colon \right.\\ &\qquad \quad \left.\frac{\xi_{n}^{\otimes m}+\lambda \tau_n}{1+\lambda}=\sigma_{n}^{\otimes m}, \ \sigma_n \in \ff\right\}\\
    &=(1+\lscrm)\sup_{\sigma\in\ff}p_{succ}(\{p_i,\Psi_i\}_i,\{M_i\}_i,\sigma^{\otimes m})   .        
        \end{aligned}
    \end{equation*}
    It follows that
$   \dfrac{p_{succ}(\{p_i,\Psi_i\}_i,\{M_i\}_i,\rho^{\otimes m})}{\sup_{\sigma\in\ff}p_{succ}(\{p_i,\Psi_i\}_i,\{M_i\}_i,\sigma^{\otimes m})}\leq (1+\lscrm) = (1+\lscr)^m
$
and, since this holds for an arbitrary $\povm$ and channel ensemble $\{p_i,\Psi_i\}_i$, we have 
\begin{equation*}
\begin{aligned}
    &\sup_{\{M_i\}_i,\{p_i,\Psi_i\}_i}\frac{p_{succ}(\{p_i,\Psi_i\}_i,\{M_i\}_i,\rho^{\otimes m})}{\sup_{\sigma\in\ff}p_{succ}(\{p_i,\Psi_i\}_i,\{M_i\}_i,\sigma^{\otimes m})} \\ &\leq(1+\lscr)^m.
    \end{aligned}
    \end{equation*}
The above considers the case in which the same measurement strategy is used throughout for $\rho$ and free states. We can also then upper bound the more realistic scenario in which the measurement strategy can be optimized individually for $\rho$ and $\sigma\in\ff$ since this will only reduce the overall ratio of success probabilities:
\begin{equation}
\begin{aligned}
&    \sup_{\{p_i,\Psi_i\}_i}\frac{\sup_{\{M_i\}_i}p_{succ}(\{p_i,\Psi_i\}_i,\{M_i\}_i,\rho^{\otimes m})}{\sup_{\{M_i\}_i}\sup_{\sigma\in\ff}p_{succ}(\{p_i,\Psi_i\}_i,\{M_i\}_i,\sigma^{\otimes m})}\\ &\leq (1+\lscr)^m.
\end{aligned}
    \end{equation}   
    \end{proof}

In convex resource theories, the above bound is tight already for a single copy, $m=1$ \cite{robustnessofcoherence,finiteconvexrobustness,infiniteconvex}. For non-convex resource theories (in which case strong duality is not guaranteed), the above bound may not be tight, and as argued in Section~\ref{sec:TwoCopies} we require at least $m=2$ to reveal a genuine advantage of resource states. Our Theorem~\ref{theoremupperbody} shows that, even in the case of non-convex resource theories in any dimension, $\lscr$ still retains a valuable interpretation in  characterizing the advantage in channel discrimination tasks, by providing an upper bound on the maximum advantage that can be achieved. To the best of our knowledge, this interesting result was not explored or presented in previous literature on non-convex resource theories \cite{finitenonconvex,longfinitenonconvex,robustnesscontinuity,starresourcetheories}

\subsubsection{Worst case advantage}
Here we consider an alternative operational interpretation for the generalized robustness, which originates from partitioning the non-convex free set $\ff$ into a collection of convex subsets, and analyzing advantages in each of the corresponding convex subtheories. Let $\ff=\bigcup_k \ff_k$ be a decomposition of $\ff$ into a union of convex sets $\ff_k$. In the finite-dimensional case, $\rob$ has been shown to quantify the worst case maximal advantage given by $\rho$ when compared to $\sigma\in\ff_k$ in single-copy channel discrimination tasks \cite{finitenonconvex,longfinitenonconvex}. The aim of this Section is to investigate conditions on $\ff$ for which this result also holds for $\lscr$ in the infinite-dimensional case. Whilst it is not necessarily true in full generality, we show it does hold true in most cases of interest.

We begin by looking at conditions under which the non-convex robustness and the worst case convex robustness are equal.

\begin{proposition}\label{lscr=inflscr} Let $\ff=\bigcup_k \ff_k$, where the each $k$ belongs to an arbitrary (not necessarily countable) set of indices, and each $\ff_k$ is closed and convex. Let $\rho\in\ddh\backslash\ff$. Then
     \begin{equation} 
    \lscr = \inf_k\underline{\mathcal{R}_{\mathcal{F}_k}}(\rho)
    \end{equation} 
    provided one of the following conditions is true:
    \begin{enumerate}
        \item $\lscr = \rob$
        \item The decomposition $\ff = \bigcup_k\ff_k$ consists of a finite number of convex sets
    \end{enumerate}
\end{proposition}

\begin{proof}
    Firstly, note that the following inequality always holds:
    \begin{equation} 
    \lscr \leq \inf_k\underline{\mathcal{R}_{\mathcal{F}_k}}(\rho).
    \end{equation}
    This is because, since $\ff_k \subseteq \ff$, we have $\lscr \leq \underline{\mathcal{R}_{\mathcal{F}_k}}(\rho)$ for all $k$. Taking infimums over $k$ on both sides yields $\lscr \leq \inf_k\underline{\mathcal{R}_{\mathcal{F}_k}}(\rho)$ as required.

    For the reverse inequality, we prove each case individually.

    {\it Case 1.}~%Let $\lscr=\rob$. From the definition of $\rob$, we have that $\forall\varepsilon > 0$ there exists some $\sigma_{\varepsilon} \in \ff$ such that $\rho \leq (1+\rob +\varepsilon)\sigma_{\varepsilon}$. Since $\ff = \bigcup_k \ff_k$, it must be that $\sigma_{\varepsilon}\in \ff_k$ for some $k(\varepsilon)$. When requiring $\sigma_{\varepsilon}\in\ff_k$, $\mathcal{R}_{\mathcal{F}_k}(\rho)$ is the minimal value such that $\rho \leq (1+\mathcal{R}_{\mathcal{F}_k}(\rho) +\varepsilon)\sigma_{\varepsilon}$. It is clear from this that $\mathcal{R}_{\mathcal{F}_k}(\rho)\leq \rob$, and taking the infimum over $k$ on both sides gives $\inf_k \underline{\mathcal{R}_{\mathcal{F}_k}}(\rho) \leq \inf_k \mathcal{R}_{\mathcal{F}_k}(\rho) \leq \rob = \lscr$ as required.
Let $\lscr=\rob$ and assume, towards a contradiction, that $\rob<\inf_k\mathcal{R}_{\ff_k}(\rho)$. This implies $\exists \ \delta>0$ such that $\rob+\delta<\inf_k\mathcal{R}_{\ff_k}(\rho)$.
 From the definition of $\rob$, we have that $\forall \varepsilon>0$ there exists some $\sigma_{\varepsilon}$ such that $\rho\leq (1+\rob+\varepsilon)\sigma_{\varepsilon}$, where each $\sigma_{\varepsilon}\in\ff_{k_\varepsilon}$ for some ($\varepsilon$-dependent) $k$. Fix some $\varepsilon<\delta/2$, then for this value of $\varepsilon$, we have 
 \begin{equation*}
     \begin{aligned} \rho&\leq\left(1+\rob+\varepsilon\right)\sigma_{\varepsilon}\\&<\left(1+\rob+\delta/2\right)\sigma_{\varepsilon}\\&<\left(1+\inf_k\mathcal{R}_{\ff_k}(\rho)-\delta/2\right)\sigma_{\varepsilon}
     \end{aligned}
\end{equation*}
 where $\sigma_{\varepsilon}$ is in some fixed $\ff_k$. This implies $\mathcal{R}_{\ff_k}(\rho)\leq \inf_k\mathcal{R}_{\ff_k}(\rho)-\delta/2<\inf_k\mathcal{R}_{\ff_k}(\rho)$, which is a contradiction. We therefore must have $\rob\geq\inf_k\mathcal{R}_{\ff_k}(\rho)$. Since by definition $\underline{\mathcal{R}_{\ff_k}}(\rho)\leq \mathcal{R}_{\ff_k}(\rho)$ for fixed $k$ (and hence $\inf_k\underline{\mathcal{R}_{\ff_k}}(\rho)\leq \inf_k\mathcal{R}_{\ff_k}(\rho)$), we also have $\rob\geq\inf_k\underline{\mathcal{R}_{\ff_k}}(\rho)$, as required.

        {\it Case 2.}~We have $\forall \varepsilon >0$, there exist sequences $\xiseq \xrightarrow{\textup{tn}}\rho$ and $\{\sigma(\varepsilon)_n\}_n \in \ff$ such that $\xi_n \leq (1+ \lscr + \varepsilon)\sigma_n$. For $\ff = \bigcup_k\ff_k$, since the union of convex sets is finite, there will always exist (at least) one of these convex sets, say $\ff_{\tilde{k}}$, that contains a subsequence $\{\sigma_{n_j}\}_{n_j}$ of $\{\sigma_n\}_{n}$ for fixed $\varepsilon$.
    This means that we can find sequences such that $\xi_{n_j} \leq (1+ \underline{\mathcal{R}_{\mathcal{F}_{\tilde{k}}}}(\rho) + \varepsilon)\sigma_{n_j}$ for any $\varepsilon>0$, and therefore $\lscr$ acts as a suboptimal value of $\underline{\mathcal{R}_{\mathcal{F}_{\tilde{k}}}}(\rho)$. This implies $\lscr \geq \underline{\mathcal{R}_{\mathcal{F}_{\tilde{k}}}}(\rho)$ and hence $\lscr \geq \inf_k\underline{\mathcal{R}_{\mathcal{F}_k}}(\rho)$.
\end{proof}
The first condition covers the cases where $\ff$ is compact in the trace norm topology. This includes any energy-constrained resource theory, any resource theory with a finite number of free states, and recovers the result in finite dimensions. It also covers the case where $\cone(\ff)$ is weak* closed even when non-convex, and thus includes the resource theory of non-Gaussianity.
The second condition relates to practically relevant cases where the free states are given by multiple constraints, such as being incoherent in multiple bases \cite{finitenonconvex}, or being a collection of thermal equilibrium states at different temperatures \cite{resourceengines}. The associated resource theories arise naturally in applications such as quantum metrology and thermal engineering.
%\color{red}
%this could have more practical examples for sure
%\color{black}

Our next result establishes that, under the conditions of Proposition \ref{lscr=inflscr}, $\lscr$ quantifies the worst case maximal advantage given by a resource state in a single-copy channel discrimination task, with respect to the decomposition into convex sets.

\begin{theorem}\label{theoremin}
    When $\lscr = \inf_k\underline{\mathcal{R}_{\mathcal{F}_k}}(\rho)$, we have 
    \begin{equation*}
    \begin{aligned}\inf_k & \!\! \sup_{\{M_i\}_i,\{p_i,\Psi_i\}_i} \frac{p_{succ}(\{p_i,\Psi_i\}_i,\{M_i\}_i,\rho)}{\sup_{\sigma\in\ff_k}p_{succ}(\{p_i,\Psi_i\}_i,\{M_i\}_i,\sigma)}\\
    &= 1 + \lscr.
    \end{aligned}
    \end{equation*}
\end{theorem}
\begin{proof}
    Since each $\ff_k$ is convex, we can use the result from \cite[Theorem~1]{infiniteconvex}, and hence 
    \begin{equation} 
    \sup_{\{M_i\}_i,\{p_i,\Psi_i\}_i} \frac{p_{succ}(\{p_i,\Psi_i\}_i,\{M_i\}_i,\rho)}{\sup_{\sigma\in\ff_k}p_{succ}(\{p_i,\Psi_i\}_i,\{M_i\}_i,\sigma)} = 1 + \underline{\mathcal{R}_{\mathcal{F}_k}}(\rho).
    \end{equation}
    Taking the infimum on each side then gives 
    \begin{equation} 
    \begin{aligned}
    \inf_k & \sup_{\{M_i\}_i,\{p_i,\Psi_i\}_i} \frac{p_{succ}(\{p_i,\Psi_i\}_i,\{M_i\}_i,\rho)}{\sup_{\sigma\in\ff_k}p_{succ}(\{p_i,\Psi_i\}_i,\{M_i\}_i,\sigma)} \\
    &=1 + \inf_k\underline{\mathcal{R}_{\mathcal{F}_k}}(\rho)\\
    &= 1 + \lscr.
    \end{aligned}
    \end{equation}
\end{proof}
This generalizes the result from \cite[Theorem~4]{finitenonconvex} and shows that, in most cases of practical interest, $\lscr$ remains an exact quantifier for the advantage given in a worst case variant of channel discrimination tasks, even in the continuous variable regime and without convexity restrictions.

\subsection{Observable lower bounds to robustness}\label{sec:lowerbounds}

Here we show that one can provide experimentally accessible lower bounds to the generalized robustness from $m$-copy witness operators. 
We will exploit again the properties of the max-relative entropy (\ref{Dmax}). By further manipulating the expressions in  the proof of Proposition~\ref{prop1m1m}, we can write
\begin{equation}
    \begin{aligned}
        &(1+\lscr)^m  \\
%        &=(1+\liminf_{\xi\to\rho}\mathcal{R_F}(\xi))^m \\
%        &=\liminf_{\xi\to\rho}(1+\mathcal{R_F}(\xi))^m \\
%        &=\liminf_{\xi\to\rho}\inf_{\sigma \in \ff} \exp(m D_{max}(\xi\|\sigma))\\
        &=\liminf_{\xi\to\rho}\inf_{\sigma \in \ff}\exp(D_{\max}(\xi^{\otimes m}\|\sigma^{\otimes m}))\\
        &=\liminf_{\xi\to\rho}\inf_{\sigma \in \ff} \sup_{
       \substack{X\in \mathcal{B(H}^{\otimes m})\ :\  0\leq X \leq I, \\ \tr[\sigma^{\otimes m}X]>0\ \forall \sigma \in \ff}
       }\! \left\{\frac{\tr[\xi^{\otimes m}X]}{\tr[\sigma^{\otimes m}X]} \right\}\\
       &\geq \liminf_{\xi\to\rho} \sup_{
       \substack{X\in \mathcal{B(H}^{\otimes m})\ :\  0\leq X \leq I, \\ \tr[\sigma^{\otimes m}X]>0\ \forall \sigma \in \ff}
       }\! \left\{\frac{\tr[\xi^{\otimes m}X]}{\sup_{\sigma\in\ff}\tr[\sigma^{\otimes m}X]} \right\}\\
       &\geq \sup_{
       \substack{X\in \mathcal{B(H}^{\otimes m})\ :\  0\leq X \leq I, \\ \tr[\sigma^{\otimes m}X]>0\ \forall \sigma \in \ff}
       }\! \liminf_{\xi\to\rho}\left\{\frac{\tr[\xi^{\otimes m}X]}{\sup_{\sigma\in\ff}\tr[\sigma^{\otimes m}X]} \right\}\\
       &=\sup_{
       \substack{X\in \mathcal{B(H}^{\otimes m})\ :\  0\leq X \leq I, \\ \tr[\sigma^{\otimes m}X]>0\ \forall \sigma \in \ff}
       }\! \left\{\frac{\tr[\rho^{\otimes m}X]}{\sup_{\sigma\in\ff}\tr[\sigma^{\otimes m}X]}\right\}\\
       &\geq \tr[\rho^{\otimes m}X]\,,\quad   \begin{array}{c}{0 \leq X\in \mathcal{B(H}^{\otimes m})}, \\
       0< \tr[\sigma^{\otimes m}X] \leq 1\,\, \forall \sigma \in \ff, \end{array}
    \end{aligned}\label{obwound}
    \end{equation}
where in the second equality we use \cite[Corollary~3.46]{Mosonyi2023} and in the subsequent lines we use the max-min inequality. 
Note that an operator $X$ satisfying the conditions in the last line of (\ref{obwound}) can be defined as $X:=\mathds{1}^{\otimes m} - \frac{W}{\|W\|_{\infty}}$, where $W$ is a $m$-copy resource witness fulfilling Eq.~(\ref{witnessdef}) with the additional constraint $\tr[\sigma^{\otimes m} W] \neq \|W\|_{\infty}$ $\forall \sigma \in \ff$. 
%and the existence of a witness $W$ is in turn guaranteed by Proposition \ref{WitnessExistence} in any resource theory with a closed $\ff$, as considered in this paper. 
We have then
\begin{equation}
\label{obiwan}
\lscr \geq \left(\max \left\{0,\, - \frac{\tr[\rho^{\otimes m} W] }{\|W\|_{\infty}}\right\}\right)^\frac1m.
\end{equation}
%for any $m$-copy resource witness $W$ fulfilling Eq.~(\ref{witnessdef}).

%\tr[\sigma^{\otimes m}X]

\section{Generalized robustness of non-Gaussianity and its applications}\label{sec:NG}
In this Section, we specialize to the resource theory of non-Gaussianity. We take the set of free states to be the set of all and only Gaussian states, $\mathcal{G}$, i.e.~all states with a Gaussian Wigner distribution \cite{CVQI_and_beyond,alessiobook} (see Sec.~\ref{sec:CV}). This forms a closed \cite{gaussianQRTs} yet non-convex subset in the infinite-dimensional space $\ddh$, and therefore doesn't fit into previously studied resource-theoretic frameworks. In particular, our approach differs from the study of so-called `quantum' or `genuine' non-Gaussianity, in which one takes as free set the {\it convex hull} of the set of Gaussian states, $\overline{\mathcal{G}}:= \text{cl conv}(\mathcal{G})$ \cite{convexG1,convexG2,infiniteconvex,longinfiniteconvex,Tommoguitar}. Our goal is to identify advantageous features of any non-Gaussian state, including those which can be obtained as convex mixtures of Gaussian states and hence have a positive, yet non-Gaussian Wigner function.

Since the cone generated by Gaussian states is weak* closed \cite{attainability}, Proposition \ref{ClosedCone} implies that Definitions \ref{robDef} and \ref{lscrDef} coincide, and hence we define the lower semi-continuous generalized robustness of non-Gaussianity of a state $\rho$ as
\begin{equation}
    \underline{\mathcal{R_G}}(\rho)=\mathcal{R_G}(\rho):=\inf_{\tau\in\ddh}\left\{\lambda\in \mathbb{R}_{\geq 0} \colon \frac{\rho + \lambda \tau}{1+\lambda} = \sigma \in \mathcal{G} \right\}.
\end{equation}
We will refer to $\underline{\mathcal{R_G}}=\mathcal{R_G}$ simply as the {\em (generalized) robustness of non-Gaussianity}, and omit the underline from now on.
By virtue of Proposition~\ref{propoprops}, the robustness $\mathcal{R_G}$ is a faithful measure of non-Gaussianity and is monotonically nonincreasing under Gaussian channels, i.e., free operations mapping Gaussian states into Gaussian states.

Recalling Proposition~\ref{lscr=inflscr}, and observing that a convex set decomposition of $\mathcal{G}$ consists of an infinite number of sets, each containing a single state, we note that an intuitive operational interpretation of ${\mathcal{R_G}}(\rho)$ arising from Theorem~\ref{theoremin} is the following: given the choice between using a non-Gaussian state $\rho$ and an arbitrary (but fixed) Gaussian state $\sigma$, there will always exist a single-copy channel discrimination task for which $\rho$ performs strictly better than $\sigma$ by a factor of at least $1+{\mathcal{R_G}}(\rho)$. In formula,
       \begin{equation}\inf_{\sigma \in \mathcal{G}}  \sup_{\{M_i\}_i,\{p_i,\Psi_i\}_i} \frac{p_{succ}(\{p_i,\Psi_i\}_i,\{M_i\}_i,\rho)}{p_{succ}(\{p_i,\Psi_i\}_i,\{M_i\}_i,\sigma)} = 1 + \mathcal{R_G}(\rho).\end{equation}

Alternatively, from Theorem~\ref{theoremupperbody} the robustness of non-Gaussianity $\mathcal{R_G}(\rho)$ also upper bounds the regularized advantage of using a non-Gaussian state $\rho$ versus any Gaussian state $\sigma \in \mathcal{G}$ in multi-copy channel discrimination tasks.

In what follows, we will evaluate and analyze the robustness of non-Gaussianity for some insightful examples, illustrated in Fig.~\ref{fig:wigners}.

\begin{figure}[t]
    \centering
    \subfloat[]{\includegraphics[width=0.5\columnwidth]{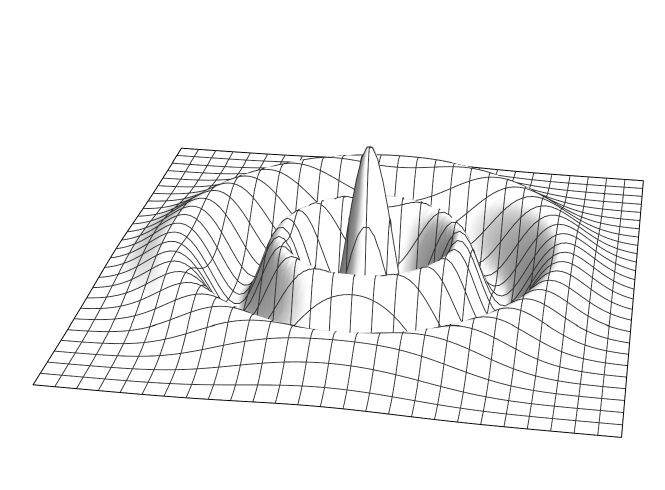}\label{fig:wignerfock}}\hfill \subfloat[]{\includegraphics[width=0.5\columnwidth]{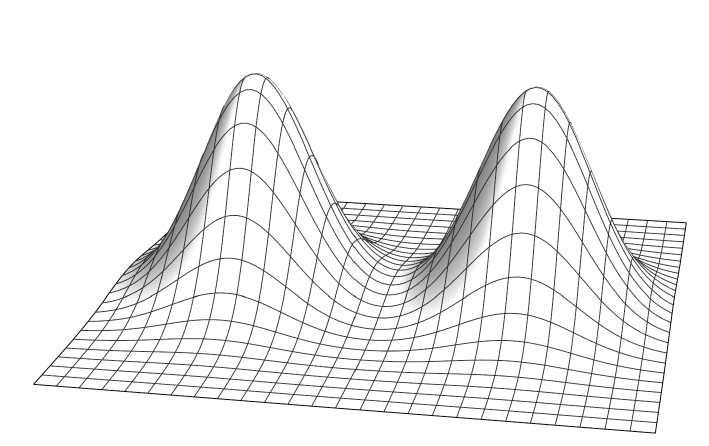}\label{fig:wignermix}}
    \caption{Illustration of the Wigner distributions for two classes of non-Gaussian states: (a)~Fock state $\ket{n}\!\bra{n}$ with $n=4$; (b)~balanced mixture $\rho_{0,d}$ of two coherent states with amplitudes $\pm d/\!\sqrt{2}$ for $d=5/2$.}
    \label{fig:wigners}
\end{figure}

\subsection{Fock states}
As a first case study, here we directly calculate $\mathcal{R_G}$ for Fock states $\ket{n}\!\bra{n}$; see Fig.~\ref{fig:wignerfock}.

We begin by recalling known results for the generalized robustness $\rob$ of an arbitrary resource. For pure states $\ket{\psi}\!\bra{\psi}$ we have \cite[Lemma~13]{longinfiniteconvex}
\begin{equation}\label{upper}
    1+\mathcal{R_F}(\ket{\psi}\!\bra{\psi})= \inf_{\sigma \in \ff} \bra{\psi}\sigma^{-1}\ket{\psi}
\end{equation}
and for arbitrary states we can use that the (Umegaki) relative entropy $D(\rho\|\sigma):=\tr[\rho (\log \rho - \log \sigma)]$ always lower bounds the max-relative entropy (\ref{Dmax}) \cite{min-maxRE}:
\begin{equation}\label{lower}
    \inf_{\sigma\in\ff}D(\rho\| \sigma) \leq \inf_{\sigma\in\ff}D_{\max}(\rho\| \sigma) =\log(\rob+1).
\end{equation}

We can then use the above bounds to calculate the robustness of non-Gaussianity of Fock states.

\begin{proposition}\label{propofocker}
    For an $n$-photon Fock state $\ket{n}\!\bra{n}$, the generalized robustness of non-Gaussianity is given by
    \begin{equation}\label{fockn}
        {\mathcal{R_G}}(\ket{n}\!\bra{n}) = \frac{(n+1)^{n+1}}{n^n}-1.
    \end{equation}
\end{proposition}

\begin{proof}
   We use the reference Gaussian state for a Fock state $\ket{n}\!\bra{n}$ (i.e. the Gaussian state with the same first and second moments \cite{extremality,gaussianHS,gaussianRE}), which is a thermal state $\tau_n$ with average photon number $n$, written as 
    \begin{equation}
   \tau_n = \frac{1}{n+1}\sum_{m=0}^{\infty}\left(\frac{n}{n+1}\right)^m\ket{m}\!\bra{m}.
   \end{equation}
   Therefore for an upper bound we use (\ref{upper}) to get 
    \begin{equation}
    {\mathcal{R_G}}(\ket{n}\!\bra{n})\leq\bra{n}\tau_n^{-1}\ket{n} -1 = \frac{(n+1)^{n+1}}{n^n}-1.
    \end{equation}

    For a lower bound, we use the fact that the relative entropy of non-Gaussianity of a state $\rho$ \cite{gaussianRE} is minimized by the Gaussian reference state $\tau_\rho$  \cite{GREmimimized}, and its value when $\rho$ is pure is given simply by the von Neumann entropy of the reference state. We hence have 
    \begin{equation}
    \begin{aligned}
    \inf_{\sigma \in \mathcal{G}}D(\ket{n}\!\bra{n}\|\ \sigma) 
    &= -\tr(\tau_n \log \tau_n)\\ 
    &= \log \left(\frac{(n+1)^{n+1}}{n^n}\right) \\
    &\leq \log({\mathcal{R_G}}(\ket{n}\!\bra{n})+1).
    \end{aligned}
    \end{equation}
    The upper and lower bounds coincide.
\end{proof}

\begin{figure}[t]
\includegraphics[width=0.95\columnwidth]{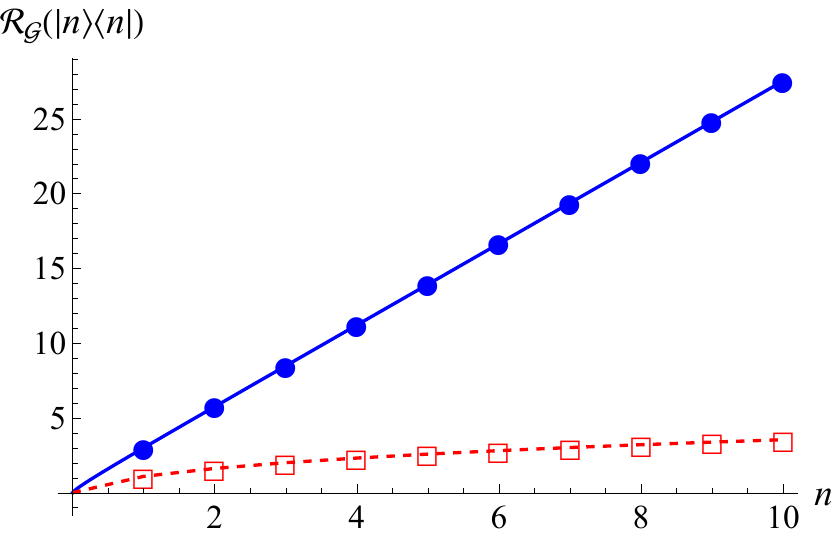}
\caption{\label{fig:rngfocks} Comparison of non-Gaussianity measures for Fock states $\ket{n}\!\bra{n}$ as a function of $n$. The solid blue line with circle markers denotes the generalized robustness of non-Gaussianity $\mathcal{R_G}(\ket{n}\!\bra{n})$ computed analytically in Eq.~(\ref{fockn}) with respect to the free set of Gaussian states ${\mathcal{G}}$. The dashed red line with square markers denotes the robustness of `genuine' non-Gaussianity $\mathcal{R}_{\overline{\mathcal{G}}}(\ket{n}\!\bra{n})$ obtained numerically in  Ref.~\cite{longinfiniteconvex} with respect to the convex hull of Gaussian states $\overline{\mathcal{G}}$.}
\end{figure}

We plot our exact expression (\ref{fockn}) for the robustness of non-Gaussianity of Fock states $\mathcal{R_G}(\ket{n}\!\bra{n})$ as a function of $n$ in Fig.~\ref{fig:rngfocks}. This is compared with the numerical expression for the robustness of `genuine' non-Gaussianity $\mathcal{R}_{\overline{\mathcal{G}}}(\ket{n}\!\bra{n})$ studied in \cite{infiniteconvex,longinfiniteconvex,note1inf} by considering the convex hull of Gaussian states $\overline{\mathcal{G}}$ as a free set. We notice that the robustness of non-Gaussianity evaluated in this paper grows linearly with $n$, $\mathcal{R_G}(\ket{n}\!\bra{n}) \approx {\rm e}\!\ n$, showing that the advantage of using Fock states in channel discrimination tasks over any Gaussian state is directly proportional to the photon number $n$. In contrast, the robustness evaluated with respect to the convex hull of Gaussian states \cite{infiniteconvex,longinfiniteconvex} grows only sublinearly, $\mathcal{R}_{\overline{\mathcal{G}}}(\ket{n}\!\bra{n}) \approx O(\sqrt{n})$.

We can also consider the multimode case where each mode $i$ has a Fock state with photon number $n_i$, with the overall state given by $\ket{\psi} = \ket{n_1} \otimes ... \otimes \ket{n_m}$. In this case, the Gaussian reference state is a tensor product of thermal states with mean photon number $n_i$ in each mode, and, using the same method as in Proposition~\ref{propofocker}, we get
\begin{equation}
    {\mathcal{R_G}}(\ket{\psi}\!\bra{\psi}) = \prod_{i=1}^m\left(\frac{(n_i+1)^{(n_i+1)}}{n_i^{n_i}}\right)-1.
\end{equation}

\subsection{Mixtures of Gaussian states}
We now consider a mixture of coherent states
\begin{equation}\label{rhomix}
\rho_{q,d} := \frac{1+q}{2} \ket{\alpha_d}\!\bra{\alpha_d} +  \frac{1-q}{2} \ket{\alpha_{-d}}\!\bra{\alpha_{-d}},
\end{equation}
with $q \in [-1,1]$, where
\begin{equation}
\ket{\alpha} = {\rm e}^{-\frac{|\alpha|^2}{2}}\sum_{n=0}^{\infty}\frac{\alpha^n}{\sqrt{n!}}\ket{n}
\end{equation}
denotes a single-mode coherent state and we take without loss of generality the coherent amplitude to be real, $\alpha_{\pm d} = \text{Re}[\alpha_{\pm d}]= \pm d/\sqrt 2$, with $d \in {\mathbb{R}}_{\geq 0}$; see Fig.~\ref{fig:wignermix}.

Note that $\rho_{q,d} \in \overline{\mathcal{G}}$, which means that the states $\rho_{q,d}$ have vanishing `genuine' non-Gaussianity, and also vanishing optical non-classicality, being mixtures of semiclassical states. Therefore, these states would be considered useless for typical resource theoretic applications constrained by convexity \cite{convexG1,convexG2,infiniteconvex,longinfiniteconvex}. Nevertheless, they are non-Gaussian, $\rho_{q,d} \not\in \mathcal{G}$, $\forall$  $d\neq 0,\ |q|\neq 1$, which means they can still be advantageous for channel discrimination following the results of this paper. Since such an advantage is quantifiable via the robustness of non-Gaussianity $\mathcal{R_G}(\rho_{q,d})$, we will provide useful bounds on the latter.

The states in (\ref{rhomix}) can equivalently be described by their Wigner distribution in phase space,
\begin{equation}\label{wignermix}
\mathcal{W}^\rho_{q,d}(x,y) = \frac{1+q}{2\pi} {\rm e}^{-(x-d)^2-y^2} + \frac{1-q}{2\pi}{\rm e}^{-(x+d)^2-y^2}.
\end{equation}
Alternatively, they can be expressed as qubit-like density operators with respect to the orthonormal basis $\{\ket{+},\ket{-}\}$, with $\ket{\pm} = \big(2(1\pm {\rm e}^{-d^2})\big)^{-1} \big(\ket{\alpha_d} \pm \ket{\alpha_{-d}}\big)$, as
\begin{equation}\label{qubitlike}
\rho_{q,d} = \frac12\left(
\begin{array}{cc}
 1+{\rm e}^{-d^2} & q \sqrt{1-{\rm e}^{-2d^2}}  \\
 q \sqrt{1-{\rm e}^{-2d^2}} q & 1-{\rm e}^{-d^2} \\
\end{array}
\right).
\end{equation}
The latter expression is useful to construct a witness operator such as the one based on Hilbert-Schmidt distance in Eq.~(\ref{lockness}), which detects non-Gaussianity in the states $\rho_{q,d}$ and enters in the explicit construction of a two-copy discrimination task for which these states outperform all Gaussian states according to Theorem~\ref{DiscriminationAdvantage}.

To provide a lower bound on the robustness, we can consider once again the relative entropy of non-Gaussianity \cite{gaussianRE,GREmimimized} and its hierarchical relation (\ref{lower}) with the max-relative entropy \cite{datta2009max},
    \begin{equation}
    \begin{aligned}
    \inf_{\sigma \in \mathcal{G}}D(\rho_{q,d}\|\sigma) 
    &= D(\rho_{q,d}\|\tau_{q,d}) \\ 
    &= S(\tau_{q,d}) - S(\rho_{q,d}) \\ 
    &\leq \log\big({\mathcal{R_G}}(\rho_{q,d})+1\big).
    \end{aligned}
    \end{equation}
Here the reference Gaussian state $\tau_{q,d}$ is a single-mode squeezed thermal state with displacement vector $\boldsymbol{\mu}_{q,d}$ and covariance matrix $\boldsymbol{V}_{q,d}$ given respectively by
\begin{equation}
\boldsymbol{\mu}_{q,d} =\left(
\begin{array}{c}
 d q \\
 0 \\
\end{array}
\right),\quad \boldsymbol{V}_{q,d}=\left(
\begin{array}{cc}
 2 d^2 \left(1-q^2\right)+1 & 0 \\
 0 & 1 \\
\end{array}
\right).
\end{equation}
The von Neumann entropy of $\tau_{q,d}$ can be evaluated via established methods \cite{holewer}, while the von Neumann entropy of $\rho_{q,d}$ can be computed by diagonalising Eq.~(\ref{qubitlike}). Explicitly,
\begin{equation}
\begin{aligned}
S(\tau_{q,d}) &= (1 + \nu_{q,d}) \log(1 + \nu_{q,d}) - \nu_{q,d} \log \nu_{q,d}\,, \\
S(\rho_{q,d}) &= - \lambda^{+}_{q,d} \log(\lambda^{+}_{q,d}) - \lambda^{-}_{q,d} \log(\lambda^{-}_{q,d})\,, 
\end{aligned}
\end{equation}
with
\begin{equation*}
\begin{aligned} 
\nu_{q,d} & = \frac{1}{2} \left(-1+\sqrt{1+2 d^2 (1- q^2)}\right), \\
\lambda^{\pm}_{q,d} & =\frac{1}{2} \left(1 \pm \sqrt{q^2+{\rm e}^{-2 d^2} \left(1-q^2\right)}\right).
\end{aligned}
\end{equation*}
We then find 
\begin{equation}\label{lowermixrelentless}
{\mathcal{R_G}}(\rho_{q,d})  \geq \exp\big[S(\tau_{q,d}) - S(\rho_{q,d})\big]-1\,.
\end{equation}
The above bound is faithful as it vanishes only for limiting values of the parameters, $|q|=1$ or $d=0$, in which cases the mixture $\rho_{q,d}$ trivially collapses into a single Gaussian state. 

A tighter bound on the robustness can be found by bounding the max-relative entropy directly in terms of the {\it measured} max-relative entropy evaluated on the classical probability distribution resulting from a POVM or projective measurement performed on the states \cite{bertasquallor}.
Consider a resource theory for which  $\lscr = \rob$ and let us denote by ${\cal M(H)}$ the set of POVMs on ${\cal H}$. We have then \cite{Mosonyi2015,Mosonyi2023,bertasquallor}
\begin{equation}
\begin{aligned}
1+ {\cal R_F}(\rho)&= \inf_{\sigma \in \ff} \sup_{k,\{M_k\}} \left\{\frac{\tr [\rho M_k]}{\tr [\sigma M_k]} : \{M_k\}_k \in {\cal M(H)}\right\} \\
&= \inf_{\sigma \in \ff} \sup_{k,\{\ket{k}\}} \left\{\frac{\bra{k}\rho\ket{k}}{\bra{k}\sigma\ket{k}} : \mbox{$\sum_k \ket{k}\!\bra{k}=\mathds{1}$}\right\} \\
&\geq  \sup_{k,\{\ket{k}\}} \inf_{\sigma \in \ff}\left\{\frac{\bra{k}\rho\ket{k}}{\bra{k}\sigma\ket{k}} : \mbox{$\sum_k \ket{k}\!\bra{k}=\mathds{1}$}\right\},
%\\&\geq  \sup_{k} \inf_{\sigma \in \ff}\left\{\frac{\bra{k}\rho\ket{k}}{\bra{k}\sigma\ket{k}} : \mbox{$\sum_k \ket{k}\!\bra{k}=\mathds{1}$}\right\},
\end{aligned}
\end{equation}
% (*Here we are restricting to measurements such that $\tr [\sigma M_x]>0$ *)
where the sum $\sum_k$ can be replaced by an integral for a continuous basis $\{\ket{k}\}$.

\begin{figure}[t]
\includegraphics[width=0.95\columnwidth]{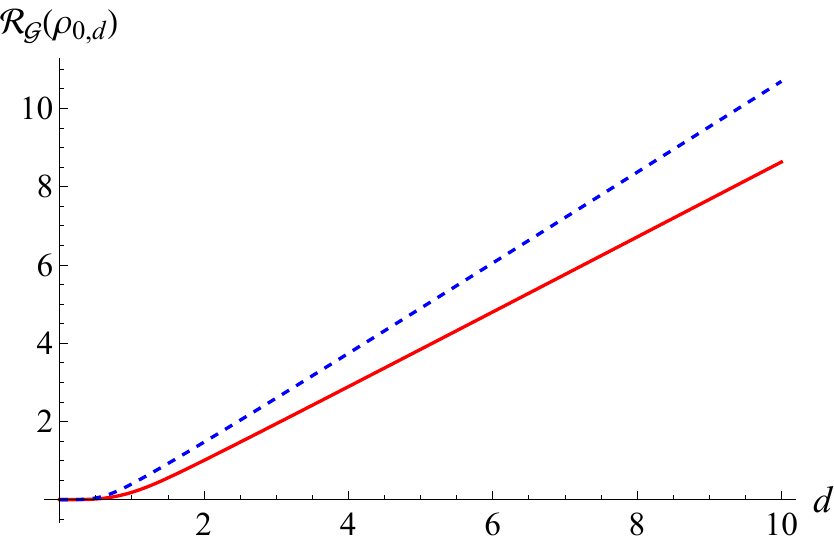}
\caption{\label{fig:rngbounds} Lower bounds on the generalized robustness of non-Gaussianity $\mathcal{R_G}$ of the states $\rho_{q,d}$ [Eq.~(\ref{rhomix})] with $q=0$, corresponding to balanced mixtures of two Gaussian coherent states with coherent amplitude $\pm d/\sqrt{2}$. The solid red curve corresponds to the bound (\ref{lowermixrelentless}) obtained from the relative entropy of non-Gaussianity. The dashed blue curve corresponds to the tighter bound (\ref{lowerhomo}) obtained from the measured max-relative entropy via an optimized quadrature measurement.}
\end{figure}

In our example, we can consider a homodyne measurement in the position basis $\{\ket{x}\}$, leading to a probability distribution given by the marginal $\bra{x} \rho \ket{x} = \int_{-\infty}^\infty \mathcal{W}^\rho(x,y) dy$ of the Wigner distribution $\mathcal{W}^\rho(x,y)$ for a single-mode state $\rho$.
We can thus write
\begin{equation}\label{wigratio}
    \begin{aligned}
       1+ {\cal R_G}(\rho_{q,d}) \geq \sup_x \inf_{\boldsymbol{\mu},\boldsymbol{V}}\frac{\int_{-\infty}^\infty \mathcal{W}^\rho_{q,d}(x,y) dy}{\int_{-\infty}^\infty \mathcal{W}^\sigma_{\boldsymbol{\mu},\boldsymbol{V}}(x,y) dy},
    \end{aligned}
\end{equation}
where $\mathcal{W}^\rho_{q,d}$ is given in (\ref{wignermix})
and $\mathcal{W}^\sigma_{\boldsymbol{\mu},\boldsymbol{V}}$ is the Wigner function (\ref{gaussianwigner}) of a generic single-mode Gaussian state $\sigma \in \mathcal{G}$ with displacement vector $\boldsymbol{\mu}$ and covariance matrix $\boldsymbol{V}$. 

From now on, we specialize to the case of a balanced cat-like mixture, $q=0$. Given the symmetry of the state $\rho_{0,d}$, we can restrict the optimization over the position eigenstates $\{\ket{x}\}$ to the positive semiaxis with no loss of generality, and we find that the optimization over the Gaussian state $\sigma \in \mathcal{G}$ is solved by $\boldsymbol{\mu} = (0,0)^T$ and $\boldsymbol{V}=\mathrm{diag}(a,1)$ with $a\geq 1$. This gives
\begin{equation}\label{envelope}
       1+ {\cal R_G}(\rho_{0,d}) \geq \sup_{x \geq 0} \inf_{a \geq 1}\frac{{\rm e}^{-(d+x)^2} \left({\rm e}^{4 d x}+1\right)/\sqrt2}{{\rm e}^{-\frac{x^2}{a}}/\sqrt{a}}.
\end{equation}
Evaluating the right-hand side in (\ref{envelope}) amounts to finding the tightest Gaussian envelope to the even Gaussian-mixture distribution appearing in the numerator. The minimization over $a$ is solved analytically by $a = 2x^2$, leading to the lower bound
\begin{equation}\label{lowerhomo}
       {\cal R_G}(\rho_{0,d}) \geq \sup_{x \geq 0} \left[{\rm e}^{\frac{1}{2}-(d+x)^2} \left({\rm e}^{4 d x}+1\right)\frac{x}{\sqrt2}\right]-1.
\end{equation}
 The value of $x$ that solves the remaining maximization, say $x = x^{\text{opt}}_d$, can be evaluated numerically. An analytical approximation  yields 
 $2 ({x^{\text{opt}}_d})^2 \geq 2 d^2+1+\tanh^8\big(\!\sqrt{d}\,\big) \geq 1$ $\forall d\geq 0$, which confirms that the optimal Gaussian state  $\sigma$ featured in the bound (\ref{wigratio}) is physical. 
 
 A comparison between the two lower bounds (\ref{lowermixrelentless}) and (\ref{lowerhomo}) to $\mathcal{R_G}(\rho_{0,d})$ is presented in Fig.~\ref{fig:rngbounds}. We see that the robustness of non-Gaussianity of $\rho_{0,d}$ (and hence the channel discrimination advantage enabled by these states) grows at least linearly with $d$, namely ${\cal R_G}(\rho_{0,d}) \gtrsim \sqrt{{\rm e}/2}\ d - 1$ for $d \gg 0$, despite the fact that these states are simple convex mixtures of Gaussian coherent states, with both a positive Wigner function and a positive Glauber-Sudarshan $P$-representation.

\section{Conclusions}
In this paper, we showed that \emph{any} resource state in \emph{any} resource theory can provide an advantage in the task of channel discrimination, without restricting to convex sets of free states or to finite dimensions. 

We began by considering multi-copy witness operators and showing that, for an arbitrary resource state, a two-copy witness operator always exists. Based on the existence of such a witness, we constructed a discrimination task to show that every resource state is advantageous over all free states in some two-copy channel discrimination task. 

Furthermore, we extended the definition of lower semicontinuous generalized robustness to the case of non-convex sets of free states in infinite-dimensional theories, and proved that it remains a faithful monotone. We investigated conditions under which this definition, when used in the general case, coincides with the simpler definition of generalized robustness employed in finite dimensions. We provided two operational interpretations for the generalized robustness monotone: firstly, we showed that it can always be used to find an upper bound for the maximal advantage given by a resource state compared to all free states in some multi-copy channel discrimination task. Secondly, we showed that, in many resource theories of interest, it exactly quantifies the worst case advantage given by a resource state compared to each set in a decomposition of the free states into convex sets. Our methods also allowed us to provide experimentally observable lower bounds for the robustness, based on the measurement of suitable multi-copy witness operators. 

Finally, we specialized to the resource theory of non-Gaussianity. Here, we calculated the generalized robustness of non-Gaussianity for Fock states with arbitrary photon number, and compared it to the robustness with respect to the convex hull of Gaussian states. We also investigated the generalized robustness of a mixture of two coherent states and provided lower bounds, thus demonstrating that even simple classical-like mixtures of Gaussian states can provide an advantage over all Gaussian states in the context of channel discrimination.

This work overcomes limitations of previous studies \cite{finiteconvexrobustness,infiniteconvex,longinfiniteconvex,finitenonconvex,longfinitenonconvex} and provides a universal operational scenario to characterize the usefulness of quantum resources. Our work demonstrates in particular that any instance of non-Gaussianity can be advantageous for specific discrimination tasks in continuous variable quantum technologies. A more systematic analysis of the power of non-Gaussian states and operations to provide enhancements in quantum sensing and metrology applications, inspired by the methods presented in this work, deserves further consideration.

\begin{acknowledgments}
We acknowledge fruitful discussions with M.~Berta, K.~Kuroiwa, M.~Mosonyi, M.~Piani, B.~Regula, S.~Sibilia, R.~Takagi, and especially L.~Lami. 
This work was supported by the Engineering and Physical Sciences Research Council (EPSRC) [Grants No.~EP/W524402/1, EP/X010929/1, and EP/T022140/1].
The authors confirm that no new data were created during this study.
\end{acknowledgments}

%\newpage
\bibliographystyle{apsrevfixedwithtitles}
\bibliography{bib} 

\begin{thebibliography}{76}
\bibitem [{not()}]{note1inf}%
  \BibitemOpen
  \href@noop {} {}\bibinfo {note} {Note, this definition is shifted by $1$ compared to \cite{infiniteconvex} where it was first introduced.}\BibitemShut {Stop}%
\bibitem [{seq()}]{sequential}%
  \BibitemOpen
  \href@noop {} {}\bibinfo {note} {Note, compactness and sequential compactness are not necessarily the same, however they are equivalent on metrizable spaces. The weak* topology is not itself induced by a metric, and hence the Banach-Alaoglu theorem only guarantees compactness. However, since $\mathcal{H}$ is separable and $\tilde{\ff}\subseteq \tth$ is compact, we can conclude $\tilde{\ff}$ is also sequentially compact, see discussion in Ref. \cite[Remark~1]{attainability} for details.}\BibitemShut {Stop}%
\end{thebibliography}
\end{document}